\documentclass[12pt, draftclsnofoot, onecolumn]{IEEEtran}
\usepackage{epsfig,latexsym}
\usepackage{float}
\usepackage{indentfirst}
\usepackage{amsmath}
\usepackage{bm}
\usepackage{amssymb}
\usepackage{times}

\usepackage{algorithm}
\usepackage[noend]{algpseudocode}
\usepackage{subfigure}
\usepackage{psfrag}
\usepackage{hyperref}
\usepackage{cite}
\usepackage{lastpage}
\usepackage{fancyhdr}
\usepackage{color}
 \usepackage{amsthm}
\usepackage{bigints}
\newtheorem{Lemma}{Lemma}
\newtheorem{lemma}[Lemma]{$\mathbf{Lemma}$}

\newcounter{problem}
\newcounter{save@equation}
\newcounter{save@problem}
\makeatletter
\newenvironment{problem}
{\setcounter{problem}{\value{save@problem}}%
  \setcounter{save@equation}{\value{equation}}%
  \let\c@equation\c@problem
  \subequations
}
{\endsubequations
  \setcounter{save@problem}{\value{equation}}%
  \setcounter{equation}{\value{save@equation}}%
}

\begin{document}%%
\title{\vspace{-0.5em} \huge{ Joint Beam Management and Power Allocation in    THz-NOMA  Networks  }}

\author{ Zhiguo Ding, \IEEEmembership{Fellow, IEEE}  and H. Vincent Poor, \IEEEmembership{Life Fellow, IEEE}    \thanks{ 
  
\vspace{-2em}

    Z. Ding and H. V. Poor are  with the Department of
Electrical and Computer Engineering, Princeton University, Princeton, NJ 08544,
USA. Z. Ding
 is also  with the School of
Electrical and Electronic Engineering, the University of Manchester, Manchester, UK (email: \href{mailto:zhiguo.ding@manchester.ac.uk}{zhiguo.ding@manchester.ac.uk}, \href{mailto:poor@princeton.edu}{poor@princeton.edu}).

  }\vspace{-2.5em}}
 \maketitle

\vspace{-1em}
\begin{abstract}
This paper investigates how to apply non-orthogonal multiple access (NOMA)   as an add-on in terahertz (THz)  networks. In particular, prior to the implementation of NOMA,  it is assumed that there exists a legacy THz system, where    spatial beams have already been configured  to serve legacy  primary users.  The aim of this paper is to study  how these pre-configured spatial beams can be used  as  a type of bandwidth resources, on which   additional secondary users are served without degrading   the performance of the legacy primary users. A   joint beam management and power allocation problem is first formulated as   a mixed combinatorial non-convex optimization problem, and then solved by two methods with different performance-complexity tradeoffs, one based on the branch and bound method and the other based on successive convex approximation. Both analytical and simulation results are presented to illustrate the new features of beam-based resource allocation in THz-NOMA networks and also demonstrate that those pre-configured spatial beams can be employed to improve the system throughput and connectivity in a spectrally efficient manner. 
\end{abstract}\vspace{-2em}

\vspace{-1.5em} 

\section{Introduction}
Terahertz (THz) communications and non-orthogonal multiple access (NOMA) are two key enabling technologies  for the envisioned sixth generation (6G) mobile network \cite{you6g, 8766143}. On the one hand, the use of THz communications is promising because a huge amount of bandwidth in the THz spectrum is available for communications \cite{6005345, 9614341,8610080,jeffthz}. On the other hand, the use of NOMA transmission can significantly improve spectral efficiency and support massive connectivity, by encouraging intelligent spectrum cooperation among mobile users \cite{mojobabook,jsacnomaxmine}.  The two communication techniques are naturally complementary to each other. For example,   using  NOMA to improve the spectral efficiency of THz networks    is well motivated by the fact that some of the anticipated applications   of 6G, such as immersive augmented reality (AR) and virtual reality (VR) as well as wireless transmission of ultra-high definition (UHD) video, can soon make  the THz spectrum as crowded as those     sub-6G Hz bands.

In the literature, THz-NOMA has been investigated from the following
two perspectives. From the performance analysis perspective, the bit-error-rate performance of THz-NOMA has been studied in \cite{9634116}, where   sophisticated schemes for adaptive superposition coding and subspace detection have been developed. In \cite{thznoma1}, NOMA has been applied to THz networks to mitigate  beam misalignment errors, where the analytical results have been developed to show that the use of NOMA can improve the outage performance and  user connectivity simultaneously.   THz-NOMA can also be applied to cooperative communications in order to improve the coverage of mobile networks, as shown in \cite{9569475}. From the resource allocation perspective, the issue of user clustering, i.e., which THz users are to be grouped together for the implementation of NOMA, has been studied in   \cite{9115278,9257475, 9695364}.  The application of advanced machine learning methods to resource allocation in THz-NOMA networks has also been investigated in \cite{9605603}, where  intelligent reflecting surfaces have been used to reconfigure the wireless propagation environment. 

Unlike these existing works about THz-NOMA, this paper  focuses on how to use NOMA as a type of add-on in THz networks. In particular, in this paper, it is assumed that there exists a legacy THz network prior to the application of NOMA, where   spatial beams have already been configured   to serve legacy    primary users.   The aim of this paper is to investigate how to use these existing spatial beams and serve additional secondary users without degrading the performance of the legacy network. There are two motivations for the considered scenario. One motivation is that showing   NOMA can be used as an add-on for existing communication systems demonstrates the compatibility of NOMA with other communication techniques and opens the door to more applications for NOMA \cite{9693536}. The other motivation is to demonstrate that spatial beams can   be used as a type of bandwidth resources, where allocating the pre-configured  beams to secondary users in THz-NOMA networks is similar to subcarrier allocation in  orthogonal frequency division multiplexing (OFDM)   systems   \cite{793310,8170332}.  The contributions of this paper are listed as follows:
\begin{itemize}
\item The considered joint beam management and power allocation problem is first formulated as  a mixed combinatorial non-convex optimization problem, and then transformed to an equivalent but more concise matrix form, where the discrete optimization variables are eliminated  and the number  of continuous optimization variables is also reduced by using the feature of the considered communication problem. 

\item Analytical results are developed in the paper to illustrate that the spatial beams can be used as bandwidth resources, but  resource allocation based on these spatial beams is fundamentally different from methods  based on conventional resources, such as OFDM subcarriers. In particular, because the pre-configured beams are not matched to the secondary users' channels, inter-beam interference exists, which makes conventional resource allocation approaches, such as water-filling power allocation,   not applicable. 

\item The optimal performance of the considered THz-NOMA network is identified by applying the branch-and-bound (BB) method to the formulated resource allocation problem \cite{5765556,8170332}. Recall that the BB method can be viewed as a type of structured exhaustive search, which means that it yields the optimal performance but  results in  significant computational complexity. Therefore,  low-complexity suboptimal resource allocation based on successive convex approximation (SCA) is also proposed in this paper \cite{6365845,7946258}. 

\item Simulation results are also presented to illustrate some interesting features of beam based resource allocation. In particular, the use of NOMA can ensure that the overall system throughput of THz networks is   significantly improved, by using those pre-configured spatial beams.  For the special case with a single secondary user, greedy scheduling, i.e., using a single beam, is optimal.   In addition, the   throughput  of the THz-NOMA network is improved if there are more secondary users involved because of  multi-user diversity, but is degraded if the number of beams is increased because of inter-beam interference.  %Furthermore, the performance gain of THz-NOMA is improved if the beams are less directional, because the beams are not designed to match the secondary users' channels.  
\end{itemize}

 \section{System Model}
 For the considered THz communication system, there are two types of users, namely primary users and secondary users. The primary users form a legacy network, and the aim of the paper is to investigate how to serve those secondary users by using the spatial beams pre-configured for the primary users, as described in the following subsections.

 \subsection{A Legacy THz Network Based on Hybrid Beamforming}
 In the considered    legacy network,    $K$ single-antenna primary users, denoted by ${\rm U}_k^P$, $1\leq k \leq K$, are served by a  base station equipped with $N$ antennas, where $K\leq N$ is assumed. Denote $s_k^P$ by the signal to be sent to primary user ${\rm U}_k^P$, and $\rho_k^P$ by the corresponding transmit power.
 
Hybrid beamforming is used to serve these $K$ primary  users, i.e., the following signal vector is sent by the base station:
  \begin{align}
\tilde{\mathbf{s}}^P= \begin{bmatrix}
\tilde{ \mathbf{f}}_1  &\cdots &\tilde{\mathbf{f}}_K
 \end{bmatrix} \mathbf{P} \mathbf{s}^P,
 \end{align}  
 where  $ \mathbf{s}^P=\begin{bmatrix}
 s_1^P &\cdots &s_K^P
 \end{bmatrix}^T$, $\mathbf{P}$ denotes the $K\times K$ digital beamforming matrix,  and $\tilde{\mathbf{f}}_k$ denotes the  analog beamforming vector for ${\rm U}_k^P$. 
 
There have been extensive studies for the design of hybrid beamforming. For example, beamsteering can be used to design analog beamforming \cite{7400949,7561012}. In particular,  $\tilde{\mathbf{f}}_k$ can be selected from the following beamsteering codebook:
 \begin{align}
\tilde{\mathbf{f}}_k\in \left\{ \frac{1}{\sqrt{N}}\mathbf{a}\left( \frac{2\pi \times 0}{N_Q} \right), \cdots ,  \frac{1}{\sqrt{N}}\mathbf{a}\left( \frac{2\pi (N_Q-1)}{N_Q} \right) \right\},
 \end{align}
where $  N_Q$ denotes the size of the   codebook,    $\mathbf{a}(\theta) $ is an $N\times 1$ vector defined as follows: 
 \begin{align}
\mathbf{a}(\theta) = \begin{bmatrix}
 1& e^{-j\frac{2\pi f_c d \sin(\theta)}{c}}&\cdots &e^{-j\frac{2(N-1)\pi f_c d \sin(\theta)}{c}} 
 \end{bmatrix}^T,
 \end{align}
 $f_c$ denotes the carrier frequency, $d$ denotes the antenna spacing, and $c$ denotes the speed of light. Therefore,  $\tilde{\mathbf{f}}_k$ can be obtained by  finding a   vector from the  codebook whose $\theta$ is closest to   ${\rm U}_k^P$'s  angle of departure. After analog beamforming is obtained,  simple approaches, such as zero forcing, can be used to design digital beamforming. % i.e., $\mathbf{P}$ can be obtained as follows:
%\begin{align}
%\mathbf{P} = \left(\mathbf{H}\mathbf{F}\right)^{-1}\mathbf{D},
%\end{align}
%where $\mathbf{F}=\begin{bmatrix}\tilde{\mathbf{f}}_1&\cdots &\tilde{\mathbf{f}}_K \end{bmatrix}$, and $\mathbf{D}$ is a diagonal matrix to ensure power normalization, i.e., ${\rm trace}(\mathbf{P}\mathbf{P}^H)=1$. 
 
In this paper, it is assumed that  both $\tilde{\mathbf{f}}_k$ and $\mathbf{P}$ have already been configured prior to the implementation of NOMA,   and the aim of the paper is to investigate how to  serve additional users without changing the configuration of the legacy network, as discussed in the next subsection.

\subsection{Serving Additional Users via THz-NOMA}
Consider   that there are   $M$ secondary users to be served via THz-NOMA. Denote the composite beamforming vector designed for primary user ${\rm U}_k^P$ by $\mathbf{f}_k$, i.e.,  
\begin{align}
\mathbf{f}_k =  \begin{bmatrix}
\tilde{ \mathbf{f}}_1  &\cdots &\tilde{\mathbf{f}}_K
 \end{bmatrix} \mathbf{p}_k ,
\end{align}
where $\mathbf{p}_k$ denotes the $k$-th column of $\mathbf{P}$.

Each of the pre-configured $K$ beams, $\mathbf{f}_k$, can be viewed as a type of bandwidth resources, and are to be allocated to the secondary users, which is similar to conventional subcarrier allocation problems in OFDMA systems.  
To facilitate the problem formulation, the beam allocation indicator, denoted by $s_{jk}$, is introduced \cite{793310,8648498}. In particular,   $s_{jk}=1$ if   secondary user ${\rm U}_j^S$ is allocated to the beam designed for primary user ${\rm U}_k^P$, otherwise $s_{jk}=0$. In order to reduce the system complexity, it is assumed that at most one secondary user can be scheduled on each of the existing beams, $\mathbf{f}_k$, which imposes the following constraints on $s_{jk}$:
\begin{align}
\sum^{M}_{j=1}s_{jk} = 1, \quad \&\quad s_{jk}\in \{0, 1\},\quad  \forall j, k.
\end{align}

By using the  beam allocation indicator,   ${\rm U}^P_k$  receives the following signal: \cite{jeffthz, 7400949} 
  \begin{align}
 y^P_k =&  \frac{a_k^P}{ \sqrt{{\rm PL}_k^P} } \mathbf{a}^H(\theta_k^P)
\sum^{K}_{i=1}\mathbf{f}_i\left(
 \sqrt{\rho_{i}^P}s_{i}^P+ \sum^{M}_{m=1}s_{mi}\sqrt{\rho_{mi}^S}s_{mi}^S\right) 
 +n_k^P,
   \end{align}  
   where ${\rm PL}_k^P$ denotes the path loss suffered by ${\rm U}_{k}^P$ and is defined as follows:
   \begin{align}
{\rm PL}_k^P=\left(\frac{c}{4\pi f_c}\right)^{-2} {e^{\zeta r_{k,P}}}{\left(r_{k,P}^{\alpha_{\rm PL}}+1\right)},
\end{align}
$r_{k,P}$ denotes the distance between the base station and ${\rm U}_{k}^P$,  $\alpha_{\rm PL}$ denotes the path loss exponent, $\zeta$ denotes the molecular absorption coefficient, $a_k^P$ denotes the fading coefficient, $\theta_k^P$ denotes ${\rm U}_k^P$'s  angle of departure, $s_{mi}^S$ denotes ${\rm U}^S_{m}$'s signal sent on ${\rm U}^P_i$'s beam,   $\rho_{mi}^S$ denotes ${\rm U}^S_{m}$'s transmit power for signal  $s_{mi}^S$, and $n_k^P$ denotes the additive white Gaussian noise with power  $\sigma^2$.

To avoid changing  the legacy system, it is assumed that the primary users treat the secondary users' signals as noise, and directly decode their own information, which means that the following data rate is achievable at the $k$-th primary user:
 \begin{align}
R_{k}^P =& \log\left( 1+\frac{ \frac{|a_k^P|^2}{ {{\rm PL}_k^P}}  |\mathbf{a}^H(\theta_k^P)\mathbf{f}_k|^2 \rho_k^P}{\frac{|a_k^P|^2}{ {{\rm PL}_k^P}}|  \mathbf{a}^H(\theta_k^P)\mathbf{f}_k|^2\sum^{M}_{m=1} s_{mk}\rho_{mk}^S+ I_{\rm IBI}^k+\sigma^2}
\right) ,
 \end{align}
 where $ I_{\rm IBI}$ denotes inter-beam interference and is given by
 \begin{align}
  I_{\rm IBI}^k= \frac{|a_k^P|^2}{ {{\rm PL}_k^P}}\sum^{K}_{i=1,i\neq k} | \mathbf{a}^H(\theta_k^P)\mathbf{f}_i|^2 \left(\rho_i^P+\sum^{M}_{m=1}s_{mi}\rho_{mi}^S\right).
 \end{align}
In order to guarantee ${\rm U}_k^P$'s target data rate which is denoted by $\bar{R}_k^P$,   the beam and   power allocation for the secondary users should  satisfy  the following condition:   $R_{k}^P \geq \bar{R}_k^P$.

  On the other hand, if $s_{jk}=1$, i.e., secondary user ${\rm U}^S_{j}$   is  served on $\mathbf{f}_k$, this secondary user can decode the primary user's signal with  the following data rate:
  \begin{align}\label{RJarrowk}
{R}_{j\rightarrow k}^S = \log\left( 1+\frac{ \frac{|a_{j}^S|^2}{ {{\rm PL}_{j}^S}}  |\mathbf{a}^H(\theta_{j}^S)\mathbf{f}_k|^2 \rho_k^P}{\frac{|a_{j}^S|^2}{ {{\rm PL}_{j}^S}}|  \mathbf{a}^H(\theta_{j}^S)\mathbf{f}_k|^2 \sum^{M}_{m=1}  s_{mk}\rho_{mk}^S+ I_{\rm IBI}^{jk} +\sigma^2}
\right) ,
\end{align} 
where the secondary user's channel parameters are defined similarly to those for the primary users and their definitions are omitted due to space limitations. The inter-beam interference, denoted by $ I_{\rm IBI}^{jk}$, is given by
 \begin{align}
  I_{\rm IBI}^{jk}=  \frac{|a_{j}^S|^2}{ {{\rm PL}_{j}^S}}\sum^{K}_{i=1,i\neq k} | \mathbf{a}^H(\theta_{j}^S)\mathbf{f}_i|^2 \left(\rho_i^P +\sum^{M}_{m=1} s_{mi}\rho_{mi}^S\right).
 \end{align}

%However, with the assumption $s_{jk}=1$, i.e., only ${\rm U}_j^S$ is served on beam $\mathbf{f}_k$, the data rate  expression can be simplified as shown  in \eqref{RJarrowk}. 

%   {\bf need to be feasible}

 Provided that  $s_{jk}=1$ and $  {R}_{j\rightarrow k}^S \geq \bar{R}_k^P$,  secondary user ${\rm U}_j^S$ can successfully decode primary user ${\rm U}_k^P$'s signal and then decode its own signal sent on beam $\mathbf{f}_k$ with the following data rate:
   \begin{align}
  {R}_{j,k}^S = \log\left( 1+\frac{ \frac{|a_{j}^S|^2}{ {{\rm PL}_{j}^S}}  |\mathbf{a}^H(\theta_{j}^S)\mathbf{f}_k|^2 \rho_{jk}^S}{    I_{\rm IBI}^{jk} +\sigma^2}
\right) .
 \end{align}

% {\bf skj has been used twice.}

The aim of this paper is to design a joint beam management   and power allocation approach for maximizing the secondary users' sum data rate,  as formulated in the following:
 \begin{problem}\label{pb:1} 
  \begin{alignat}{2}
\underset{\rho_{jk}^S,s_{jk}}{\rm{max}} &\quad    
  \sum^{M}_{j=1} \sum^{K}_{k=1}s_{jk}  {R}_{j,k}^S\label{1tobj:1} \\
\rm{s.t.} & \quad s_{jk}\left( R_{k,j}^P  
-\bar{R}_k^P
\right)\geq 0, \forall k, j \label{1tst:1}
\\
& \quad   \tilde{R}_{k}^P -\bar{R}_k^P \geq 0 , \forall k
\label{1tst:11}\\
& \quad s_{jk}\left(  {R}_{j\rightarrow k}^S -\bar{R}_k^P\right)\geq 0 , \forall k, j  \label{1tst:2}
\\ &\quad 
 s_{jk}\in \{0,1\}, \forall k,j , \quad \sum^{M}_{j=1}s_{jk}\leq 1, \forall k\label{1tst:4}
\\
& \quad  \sum_{j=1}^{M}\sum^{K}_{k=1}s_{jk}\rho^S_{jk}\leq P^{{\rm max}} ,   \label{1tst:5}
  \end{alignat}
\end{problem} 
where $\tilde{R}_{k}^P =  \log\left( 1+\frac{ \frac{|a_k^P|^2}{ {{\rm PL}_k^P}}  |\mathbf{a}^H(\theta_k^P)\mathbf{f}_k|^2 \rho_k^P}{  I_{\rm IBI}^k+\sigma^2}
\right) $, and 
$P_{\rm max}$ denotes the transmit power budget for the secondary users. The constraints in \eqref{1tst:1} and \eqref{1tst:11}  ensure that the  primary users' quality of service (QoS) requirements, i.e., their target data rates, can be met when additional secondary users are served on the existing $K$ beams. Note that in addition to constraint   \eqref{1tst:1}, constraint  \eqref{1tst:11}  is required  because it is possible that no secondary user is scheduled on beam $\mathbf{f}_k$ but primary user ${\rm U}_k^P$ still suffers   the interference from the secondary users on the other beams.   Constraint  \eqref{1tst:11} can be omitted  if the beams are orthogonal, i.e., $\mathbf{f}_k^H\mathbf{f}_i=0$, $k\neq i$.   

 The constraint in \eqref{1tst:2} ensures that for a secondary user which is scheduled on beam $\mathbf{f}_k$,  successive interference cancellation (SIC) can be carried out successfully. The constraint in   \eqref{1tst:4} ensures that at most one secondary user is served on each of the $K$ beams. It is important to point out that it is possible that none of the secondary users is scheduled on one  beam, and one secondary user is scheduled on multiple beams, i.e., the secondary users are scheduled in an opportunistic   manner, where the investigation of user fairness among the secondary users is beyond  the scope of this paper and will be treated as a promising direction for future research.  

 Problem \ref{pb:1} is challenging to solve since it is a mixed combinatorial non-convex
optimization problem. In particular, it is straightforward to verify that   the objective function is not concave, and the constraints in \eqref{1tst:1} and \eqref{1tst:2}  are not convex. In addition, the beam allocation indicator, $s_{jk}$, is a binary optimization variable. In this paper,    problem \ref{pb:1} will be solved  by applying the BB  and SCA methods which realize different performance-complexity  tradeoffs.

% \subsection{Motivation for the Application of NOMA  }
%Assume that there are $M$ single-antenna secondary users, denoted by ${\rm U}_m^S$, $1\leq m \leq M$.  Without changing the configuration of the legacy system, the secondary users can treat the primary users' beamforming vectors as types of bandwidth resources. In particular, denote the composite beamforming vector designed for primary user ${\rm U}_k^P$ by $\mathbf{f}_k$ as defined in the following:
%\begin{align}
%\mathbf{f}_k =  \begin{bmatrix}
%\tilde{ \mathbf{f}}_1  &\cdots &\tilde{\mathbf{f}}_K
% \end{bmatrix} \mathbf{p}_k .
%\end{align}

%\subsubsection{Non-orthogonal $\mathbf{P}$}
%Full analog beamforming, we can use my previous examples. 
%\subsubsection{Orthogonal $\mathbf{P}$}

\section{Problem Reformulation }
In this section, the joint beam and power allocation problem presented in \eqref{pb:1}  will be   reformulated  to facilitate the applications of the BB and SCA methods, where the property of the considered optimization problem is also studied.  
 
In order to simplify the notations, first define  $ h_{ki}^P\triangleq \frac{|a_k^P|^2}{ {{\rm PL}_k^P}}  | \mathbf{a}^H(\theta_k^P)\mathbf{f}_i|^2$, and  $R_{k}^P$ can be rewritten as follows: 
 \begin{align}
R_{k}^P  
=&\log\left( 1+\frac{ h_{kk}^P \rho_k^P}{h_{kk}^P \sum^{M}_{m=1}s_{mk}\rho_{mk}^S+  \sum^{K}_{i=1,i\neq k} h_{ki}^P(\rho_i^P+\sum^{M}_{m=1}s_{mi}\rho_{mi}^S) +\sigma^2}
\right) .
 \end{align}
 By using this simplified expression of $R_{k}^P  $, constraint \eqref{1tst:1}  can be simplified as follows: 
  \begin{align}\label{constraint1}
s_{jk}\left(  s_{jk}\rho_{jk}^S-  \frac{   \rho_k^P}{2^{\bar{R}_k^P}-1}+\frac{1}{h_{kk}^P }\sum^{K}_{i=1,i\neq k} h_{ki}^P\left(\rho_i^P+\sum^{M}_{m=1}s_{mi}\rho_{mi}^S\right) +\frac{\sigma^2}{h_{kk}^P } \right)\leq 0,
  \end{align}
  where $ \sum^{M}_{m=1}s_{mk}\rho_{mk}^S$ is reduced to $ s_{jk}\rho_{jk}^S$ because of the use of $s_{jk}$ outside of the bracket at the left-hand side of \eqref{constraint1}. In particular, if $s_{jk}=1$, i.e., secondary user ${\rm U}_{j}^S$ is   scheduled on beam $\mathbf{f}_k$,  $ \sum^{M}_{m=1}s_{mk}\rho_{mk}^S=s_{jk}\rho_{jk}^S$. If $s_{jk}=0$, the constraint shown in \eqref{constraint1} is not active, and the expression inside of the   bracket at the left-hand side of \eqref{constraint1} has no impact.

  Note that $\rho_k^P$ is a constant because the primary users' transmit powers are assumed to be fixed. Therefore, constraint \eqref{1tst:1}  can be further simplified as follows:
\begin{align}\label{p1cnew}
 s_{jk}\left( 
s_{jk}\rho_{jk}^S+\frac{1}{h_{kk}^P }\sum^{K}_{i=1,i\neq k} h_{ki}^P\sum^{M}_{m=1}s_{mi}\rho_{mi}^S+c_k
\right)\leq 0,
\end{align}  
where     $c_k =  \frac{1}{h_{kk}^P }\sum^{K}_{i=1,i\neq k} h_{ki}^P\rho_i^P-\frac{  \rho_k^P}{2^{\bar{R}_k^P}-1}+\frac{\sigma^2}{h_{kk}^P }$.
  
Similarly by introducing    the following definition, $h_{jk}^S=\frac{|a_j^S|^2}{ {{\rm PL}_j^S}}  |\mathbf{a}^H(\theta_j^S)\mathbf{f}_k|^2$, the data rate for ${\rm U}_{j}^S$ to decode ${\rm U}_k^P$ can be simplified as follows: 
  \begin{align}
 {R}_{j\rightarrow k}^S 
=& \log\left( 1+\frac{ h_{jk}^S \rho_k^P}{h_{jk}^S  \sum^{M}_{m=1}s_{mk}\rho_{mk}^S+  \sum^{K}_{i=1,i\neq k} h_{ji}^S(\rho_i^P + \sum^{M}_{m=1}s_{mi}\rho_{mi}^S) +\sigma^2}
\right) ,
 \end{align}
 which means that   constraint  \eqref{1tst:2} which ensures the condition $  {R}_{j\rightarrow k}^S \geq \bar{R}_k^P$  can be rewritten as follows:
   \begin{align}
s_{jk}\left( s_{jk}\rho_{jk}^S -   \frac{  \rho_k^P}{2^{\bar{R}_k^P}-1}+\frac{1}{h_{jk}^S}\sum^{K}_{i=1,i\neq k} h_{ji}^S\left(\rho_i^P +\sum^{M}_{m=1}s_{mi}\rho_{mi}^S\right)  +\frac{\sigma^2}{h_{jk}^S}\right)\leq 0,
 \end{align}
 where  the beam allocation indicator is used to simplify the term $ \sum^{M}_{m=1}s_{mk}\rho_{mk}^S$ to $  s_{jk}\rho_{jk}^S$. 
By using the fact that the primary users' powers are constants, constraint  \eqref{1tst:2} can be further simplified as follows:
    \begin{align}\label{p1dnew}
s_{jk}\left( s_{jk}\rho_{jk}^S +\frac{1}{h_{jk}^S}\sum^{K}_{i=1,i\neq k} h_{ji}^S \sum^{M}_{m=1}s_{mi}\rho_{mi}^S   +b_{jk}\right)\leq 0,
 \end{align}
 where $b_{jk} =\frac{1}{h_{jk}^S}\sum^{K}_{i=1,i\neq k} h_{ji}^S \rho_i^P  -   \frac{  \rho_k^P}{2^{\bar{R}_k^P}-1}+ \frac{\sigma^2}{h_{jk}^S}$.

 Similarly, by 
   applying the above reformulation steps, the objective function can be also simplified    as follows:
   \begin{align}\nonumber
  {R}_{j,k}^S = &\log\left( 1+\frac{ h_{jk}^S   \rho_{jk}^S}{ \sum^{K}_{i=1,i\neq k} h_{ji}^S(\rho_i^P +\sum^{M}_{m=1}s_{mi}\rho_{mi}^S)   +\sigma^2}
\right)\\\label{objnew}
=&
\log\left( 1+\frac{ h_{jk}^S   \rho_{jk}^S}{ \sum^{K}_{i=1,i\neq k} h_{ji}^S \sum^{M}_{m=1}s_{mi}\rho_{mi}^S  +t_{jk}}
\right),
 \end{align}
 where $t_{jk} =  \sum^{K}_{i=1,i\neq k} h_{ji}^S\rho_i^P   +\sigma^2$. 
 
 By applying  \eqref{p1cnew}, \eqref{p1dnew} and \eqref{objnew} to problem \ref{pb:1}, the considered optimization  problem can   be equivalently recast as follows: 
   \begin{problem}\label{pb:2} 
  \begin{alignat}{2}
\underset{\rho_{jk}^S,s_{jk}}{\rm{max}} &\quad    
  \sum^{M}_{j=1}   \sum^{K}_{k=1} s_{jk}\log\left( 1+\frac{ h_{jk}^S   \rho_{jk}^S}{ \sum^{K}_{i=1,i\neq k} h_{ji}^S \sum^{M}_{m=1}s_{mi}\rho_{mi}^S  +t_{jk}}
\right) \label{2tobj:1} \\
\rm{s.t.} & \quad    s_{jk}\left( 
s_{jk}\rho_{jk}^S+\frac{1}{h_{kk}^P }\sum^{K}_{i=1,i\neq k} h_{ki}^P\sum^{M}_{m=1}s_{mi}\rho_{mi}^S+c_k
\right)\leq 0, \forall k, j\label{2tst:1}
\\&\quad
 \frac{1}{h_{kk}^P }\sum^{K}_{i=1,i\neq k} h_{ki}^P\sum^{M}_{m=1}s_{mi}\rho_{mi}^S+c_k
 \leq 0,  \forall k\label{2tst:11}
\\ &\quad s_{jk}\left( s_{jk}\rho_{jk}^S +\frac{1}{h_{jk}^S}\sum^{K}_{i=1,i\neq k} h_{ji}^S \sum^{M}_{m=1}s_{mi}\rho_{mi}^S   +b_{jk}\right)\leq 0 , \forall k, j \label{2tst:2}
\\
&\quad  \eqref{1tst:1}, \eqref{1tst:4}, \& \eqref{1tst:5}.
  \end{alignat}
\end{problem} 
Problem \ref{pb:2} is concise enough to obtain certain insight for the feature of beam-based resource allocation, as shown in the next subsection. 

\subsection{Special Cases with $M=1$ and $K>1$} When $M=1$, there is a single secondary user and  problem \ref{pb:2}  is to find out how the overall transmit power, $P^{\max}$, can be distributed among the $K$ beams. Intuitively, the water-filling approach can be applied, i.e., all the beams are employed and more power is allocated to a beam with a stronger channel gain. However, our conducted simulation results show a surprising result that greedy scheduling, i.e., using a single beam, is optimal. In particular, the greedy scheduling problem can be simply formulated as follows:
  \begin{problem}\label{pb:2y} 
  \begin{alignat}{2}
\underset{\rho_{1k}}{\rm{max}} &\quad    
   \log\left( 1+\frac{ h_{1k}^S   }{  t_{1k}} \max\{0,\min\{P^{\max}, -c_k, -b_{1k}\}\}
\right) \label{2ytobj:1} .
  \end{alignat}
\end{problem}

The following lemma shows that the conjecture that greedy scheduling is optimal   holds in the special case with $M=1$ and $K=2$. 

\begin{lemma}\label{lemma0}
Consider a special case with $M=1$ and $K=2$, where  the legacy network has been design to ensure that   the   target data rates of the primary users are     small, i.e., $\bar{R}_k^P\rightarrow 0$,  $k\in\{1,2\}$, and   all the primary users use the same transmit power.  At high SNR, i.e., $\sigma^2\rightarrow0$, the optimal solution of problem \ref{pb:2} is the same as that of problem \ref{pb:2y}. 
\end{lemma}
\begin{proof}
See Appendix \ref{proof0}. 
\end{proof}
%$\rho^S_{1k^*}=P^{\max}$, where $k^*=\arg \max \{h_{1k}^S, 1\leq k \leq 2\}$

The proof for the conclusion that greedy scheduling is optimal for  a more general case with $M=1$ and $K>1$ is difficult to obtain, and will be an important direction for future research.  
Note that for the general cases with $M>1$, greedy scheduling is not optimal, and problem \ref{pb:2} needs to be further rearranged to facilitate the application of the BB and SCA methods, as shown in the next subsections.

\subsection{Eliminating the Binary Optimization Variables,  $s_{jk}$  } Problem \ref{pb:2} is challenging to solve due to the facts that $s_{jk}$ is binary and also   the two optimization variables, $s_{jk}$ and $\rho^S_{jk}$, are coupled. As shown in \cite{8648498}, the binary optimization variables can be eliminated    by introducing   the following continuous variable:  $\tilde{\rho}_{jk}^S=s_{jk}\rho^S_{jk}$. By using this auxiliary variable, constraints \eqref{2tst:1} and \eqref{2tst:11}  can be combined together and equivalently expressed as follows:
\begin{align}\label{new1}
 \tilde{\rho}_{jk}^S+\frac{1}{h_{kk}^P }\sum^{K}_{i=1,i\neq k} h_{ki}^P\sum^{M}_{m=1}\tilde{\rho}_{mi}^S+c_k
\leq 0,\forall k, j,
\end{align}
which can be explained in the following.  For the case that  $s_{jk}=1$,   $\tilde{\rho}_{jk}^S=\tilde{\rho}_{jk}^S$, and it is straightforward to show 
that \eqref{new1} is equivalent to \eqref{2tst:1}, which is  stricter   than \eqref{2tst:11} and hence constraint \eqref{2tst:11} can be ignored in this case.  For the case  that  $s_{jk}=0$, $\tilde{\rho}_{jk}^S=0$,  \eqref{new1} is the same as \eqref{2tst:11}, whereas constraint  \eqref{2tst:1} is not active in this case. 

Intuitively,  constraints \eqref{2tst:1} can also  be equivalently reformulated to the following concise expression by using $\tilde{\rho}_{jk}^S$:  
\begin{align}\label{new2}
 \tilde{\rho}_{jk}^S +\frac{1}{h_{jk}^S}\sum^{K}_{i=1,i\neq k} h_{ji}^S \sum^{M}_{m=1} \tilde{\rho}_{mi}^S   +b_{jk} \leq 0 , \forall k, j.
\end{align}
For the case that $s_{jk}=1$,   \eqref{new2} is indeed equivalent to \eqref{2tst:1}. However, for the case that $s_{jk}=0$, \eqref{new2} is not equivalent to \eqref{2tst:1}, because  the original constraint in \eqref{2tst:1} is not active in this case but the new constraint in \eqref{new2} is still active and expressed as follows: $\frac{1}{h_{jk}^S}\sum^{K}_{i=1,i\neq k} h_{ji}^S \sum^{M}_{m=1} \tilde{\rho}_{mi}^S   +b_{jk} \leq 0$. Or in other words, the constraint in \eqref{new2} cannot be used to replace \eqref{2tst:1} because an   extra constraint is introduced    if secondary user ${\rm U}^S_j$ is not scheduled on beam $\mathbf{f}_k$. Instead, constraint \eqref{2tst:1}  can be equivalently recast as follows:
\begin{align}\label{new3}
{\rm sign} (\tilde{\rho}_{jk}) \left( \tilde{\rho}_{jk}^S +\frac{1}{h_{jk}^S}\sum^{K}_{i=1,i\neq k} h_{ji}^S \sum^{M}_{m=1} \tilde{\rho}_{mi}^S   +b_{jk} \right)\leq 0 , \forall k, j,
\end{align}
where ${\rm sign}(x)$ denotes the sign of $x$.

Furthermore, $\tilde{\rho}_{jk}$ can also be used to simplify the objective function, where
the following equality can be established:
\begin{align}\label{eq21}
s_{jk}\log\left( 1+\frac{ h_{jk}^S   \rho_{jk}^S}{ \sum^{K}_{i=1,i\neq k} h_{ji}^S \sum^{M}_{m=1}s_{mi}\rho_{mi}^S  +t_{jk}}
\right)  =  \log\left( 1+\frac{ h_{jk}^S  \tilde{ \rho}_{jk}^S}{ \sum^{K}_{i=1,i\neq k} h_{ji}^S \sum^{M}_{m=1} \tilde{\rho}_{mi}^S  +t_{jk}}
\right) ,
\end{align}
which is explained in the following. For the case that $s_{jk}=1$,   $\tilde{\rho}_{jk}^S=\tilde{\rho}_{jk}^S$ and hence the two sides of \eqref{eq21} are the same. For the case that $s_{jk}=0$, the two sides of \eqref{eq21} are zero and still equivalent. 

 Note that the use of these new expressions shown in  \eqref{new1}  \eqref{new3}, and \eqref{eq21}  can avoid   using the beam assignment indicator.   However, in order to ensure that at most a single secondary user is scheduled on one beam, i.e., $\sum^{M}_{j=1}s_{jk}=1$, a penalty variable, denoted by $\xi$, needs to be introduced, and the new objective function is given by 
\begin{align}\label{new obj} 
  \sum_{j=1}^M\sum^K_{k=1}   \log\left( 1+\frac{ h_{jk}^S   \tilde{\rho}_{jk}^S}{ \xi h_{jk}^S\sum^{M}_{m=1,m\neq j}\tilde{\rho}_{mk}^S + \sum^{K}_{i=1,i\neq k} h_{ji}^S \sum^{M}_{m=1}\tilde{\rho}_{mi}^S  +t_{jk}}
\right) . 
\end{align}
By using \eqref{new1}  \eqref{new3}, and \eqref{new obj}, problem \ref{pb:2} can be recast as follows:\footnote{It can be straightforwardly verified   that  the objective and the constraints of problem \ref{pb:3} are monotonic functions, which means that similar to the BB method,  monotonic optimization can also be used to find the optimal solution of problem \ref{pb:3}.  } 
 \begin{problem}\label{pb:3} 
  \begin{alignat}{2}
\underset{\rho_{jk}^S }{\rm{max}}  &\quad    
   \sum_{j=1}^M\sum^K_{k=1}    \log\left( 1+\frac{ h_{jk}^S   \tilde{\rho}_{jk}^S}{ \xi h_{jk}^S\sum^{M}_{m=1,m\neq j}\tilde{\rho}_{mk}^S + \sum^{K}_{i=1,i\neq k} h_{ji}^S \sum^{M}_{m=1}\tilde{\rho}_{mi}^S  +t_{jk}}
\right) \label{3tobj:1} \\
\rm{s.t.} & \quad     \sum_{j=1}^M\sum^K_{k=1} \tilde{\rho}_{jk}\leq P^{\rm max}\label{3ttst:3}
\\
&\quad 
 \eqref{new1},  \eqref{new3}.
  \end{alignat}
\end{problem} 
As shown in \cite{8648498}, for the case $\xi\rightarrow \infty$, problem \ref{pb:3} is equivalent to problem \ref{pb:2}.

\subsection{Reducing the Number of Variables to Be Optimized} For problem \ref{pb:3}, there are $MK$    variables to be optimized, i.e., $\tilde{\rho}_{jk}$, $1\leq j\leq M$ and $1\leq k\leq K$,  which can cause significant computational complexity. It is important to point out that the constraints in \eqref{new1} and \eqref{new3} can be used to  reduce the number of optimization variables. For example, by using the constraint in \eqref{new1}, one can conclude that   if $c_k>0$,  $\tilde{\rho}_{jk}^S=0$ which also means $s_{jk}=0$, $\forall j\in\{1,\cdots, M\}$, because the following inequality can never be satisfied
\begin{align}
 \tilde{\rho}_{jk}^S+\frac{1}{h_{kk}^P }\sum^{K}_{i=1,i\neq k} h_{ki}^P\sum^{M}_{m=1}\tilde{\rho}_{mi}^S+c_k\leq  0.
\end{align}
This conclusion  is also expected as explained in the following. Recall that $c_k =  \frac{1}{h_{kk}^P }\sum^{K}_{i=1,i\neq k} h_{ki}^P\rho_i^P-\frac{  \rho_k^P}{2^{\bar{R}_k^P}-1}+\frac{\sigma^2}{h_{kk}^P }$, and hence   $c_k> 0$ is equivalent to the following:
\begin{align}
\log\left(1+\frac{  h_{kk}^P \rho_k^P}{\sum^{K}_{i=1,i\neq k} h_{ki}^P\rho_i^P+ \sigma^2}\right) <\bar{R}_k^P,
\end{align}
which means that ${\rm U}^P_k$'s QoS requirement cannot be satisfied even if no secondary user is served on beam $\mathbf{f}_k$. Or in other words, if $c_k>0$, beam $\mathbf{f}_k$ is not available to any secondary users.  Similarly, by using the constraint in \eqref{new3}, one can conclude that $b_{jk}>0$ leads to  $\tilde{\rho}^S_{jk}=0$ which also means $s_{jk}=0$ .

In order to use the two conclusions for reducing  the number of the optimization variables, first build the following set:
\begin{align}
\mathcal{S} = \left\{\{j,k\}\left | b_{jk} \leq 0, c_k\leq 0 \right.\right\}.
\end{align} 
Based on the previous discussions,  $\tilde{\rho}_{jk}=0$ and $s_{jk}=0$ , if $\{j,k\}\not\in\mathcal{S}$, and hence there is no need to optimize these variables. 
Therefore, by using the set $\mathcal{S}$,   problem \ref{pb:3} can be equivalently expressed  as follows: 
  \begin{problem}\label{pb:4} 
  \begin{alignat}{2}
\underset{\rho_{jk}^S, \forall\{j,k\}\in \mathcal{S}}{\rm{max}}  &\quad    
  \sum_{\forall \{j,k\}\in\mathcal{S}}    \log\left( 1+\frac{ h_{jk}^S   \tilde{\rho}_{jk}^S}{ \xi h_{jk}^S\sum^{M}_{m=1,m\neq j}\tilde{\rho}_{mk}^S + \sum^{K}_{i=1,i\neq k} h_{ji}^S \sum^{M}_{m=1}\tilde{\rho}_{mi}^S  +t_{jk}}
\right) \label{4tobj:1} \\
\rm{s.t.} & \quad    
 \tilde{\rho}_{jk}^S+\frac{1}{h_{kk}^P }\sum^{K}_{i=1,i\neq k} h_{ki}^P\sum^{M}_{m=1}\tilde{\rho}_{mi}^S+c_k
 \leq 0, \forall \{j,k\}\in\mathcal{S}\label{4ttst:1}
\\ &\quad {\rm sign} (\tilde{\rho}_{jk}) \left( \tilde{\rho}_{jk}^S +\frac{1}{h_{jk}^S}\sum^{K}_{i=1,i\neq k} h_{ji}^S \sum^{M}_{m=1} \tilde{\rho}_{mi}^S   +b_{jk} \right)\leq 0 , \forall \{j,k\}\in\mathcal{S} \label{4ttst:2}
\\
&\quad  \sum_{\forall \{j,k\}\in\mathcal{S}}\tilde{\rho}_{jk}\leq P^{\rm max}.\label{4ttst:3}
  \end{alignat}
\end{problem}

\subsection{Reformulating the Problem into a Matrix-Based Form }   In this subsection, problem \ref{pb:4} will be reformulated into a matrix form  to facilitate the application of the  BB and SCA methods.    
  Recall that the set, $\mathcal{S}$, contains the indices of the secondary users whose transmit powers can be non-zero and need to be optimized, i.e., if $\{j,k\}\in \mathcal{S}$, secondary user ${\rm U}_j^S$ can be scheduled on  beam $\mathbf{f}_k$. Denote the size of $\mathcal{S}$ by $|\mathcal{S}|$. By using $\mathcal{S}$, one can build a $|\mathcal{S}|\times 2$ matrix, denoted by $\mathbf{S}$, where  the two elements on the $i$-th row of $\mathbf{S}$ are the $i$-th element of $\mathcal{S}$. For example, if $\mathcal{S}=\{\{1,1\},\{1,2\},\{2,2\}\}$, $\mathbf{S}=\begin{bmatrix}  1&1&2\\ 1& 2&2\end{bmatrix}^T$. 
  
  Furthermore, denote the element on the $p$-th row and $q$-th column of a matrix $\mathbf{S}$ by $\mathbf{S}_{pq}$.  
 Define $y_p=\tilde{\rho}_{\mathbf{S}_{p1}\mathbf{S}_{p2}}$, and $\mathbf{y}=\begin{bmatrix}y_1 &\cdots &y_{|\mathcal{S}|} \end{bmatrix}^T$ which collects all the variables to be optimized. Define    $\tilde{\boldsymbol \rho}$ as a $KM\times 1$ vector collecting the original $KM$ variables, $\tilde{\rho}_{jk}$, i.e.,
 \begin{align}
 \tilde{\boldsymbol \rho} = \begin{bmatrix}
 \tilde{\rho}_{11}&\cdots  &\tilde{\rho}_{M1} &\cdots & \tilde{\rho}_{1K}&\cdots & \tilde{\rho}_{MK} 
 \end{bmatrix}^T.
 \end{align}
 Furthermore, define $\mathbf{R}$ as a $  KM \times |\mathcal{S}|$ mapping matrix to ensure $\tilde{\boldsymbol \rho}=\mathbf{R}\mathbf{y} $, where $\mathbf{R}$ can be built as follows.   $\mathbf{R}$ is an all zero matrix, except that   the  element on the $(M(\mathbf{S}_{p2}-1)+\mathbf{S}_{p1})$-th row and $p$-th column of $\mathbf{R}$, $1\leq p\leq  |\mathcal{S}|$, is set as one.  For the above example with $\mathcal{S}=\{\{1,1\},\{1,2\},\{2,2\}\}$, $\mathbf{y}=\begin{bmatrix}y_1 &y_2 &y_3 \end{bmatrix}^T$, $y_1=\tilde{\rho}_{\mathbf{S}_{11}\mathbf{S}_{12}}=\tilde{\rho}_{11}$, $y_1=\tilde{\rho}_{\mathbf{S}_{21}\mathbf{S}_{22}}=\tilde{\rho}_{12}$, $y_2=\tilde{\rho}_{\mathbf{S}_{31}\mathbf{S}_{32}}=\tilde{\rho}_{22}$, $ \tilde{\boldsymbol \rho} = \begin{bmatrix}
 \tilde{\rho}_{11}&   \tilde{\rho}_{21}  & \tilde{\rho}_{12}& \tilde{\rho}_{22} 
 \end{bmatrix}^T$, and $ \mathbf{R} $ is defined as follows:
 \begin{align}
 \mathbf{R} = \begin{bmatrix}  1 &0&0\\0 &0&0 \\ 0&1&0\\ 0 &0&1
 \end{bmatrix}.
 \end{align}
 
 By using the above definitions, problem \ref{pb:4} can be equivalently recast as follows:
  \begin{problem}\label{pb:5} 
  \begin{alignat}{2}
\underset{\mathbf{y}}{\rm{max}}  &\quad    
  \sum_{p=1} ^{|\mathcal{S}|}   \log\left( 1+\frac{ \mathbf{c}_p^T\mathbf{y}}{ \mathbf{d}_p^T\mathbf{R}\mathbf{y}  +t_{\mathbf{S}_{p1}\mathbf{S}_{p2}}}
\right) \label{5tobj:1} \\
\rm{s.t.} & \quad    
\mathbf{e}_p^T\mathbf{R}\mathbf{y}+c_{\mathbf{S}_{p2}}
 \leq 0, 1\leq p\leq |\mathcal{S}|\label{5ttst:1}
\\ &\quad  {\rm sign}(y_p) \left(\mathbf{f}_p ^T\mathbf{R}\mathbf{y}  +b_{\mathbf{S}_{p1}\mathbf{S}_{p2}}\right)\leq 0 , 1\leq p\leq |\mathcal{S}| \label{5ttst:2}
\\
&\quad \mathbf{1}_{|\mathcal{S}|\times 1}^T\mathbf{y}\leq P^{\rm max}.\label{5ttst:3}
  \end{alignat}
\end{problem}
where  $\mathbf{c}_p$ is a $|\mathcal{S}|\times1 $ vector as follows:
\begin{align}
\mathbf{c}_p = \begin{bmatrix}\mathbf{0}_{1\times (p-1) }&h_{\mathbf{S}_{p1}\mathbf{S}_{p2}}^S   & \mathbf{0}_{1\times (|\mathcal{S}|-p)} \end{bmatrix}^T,
\end{align}
$\mathbf{d}_p$ is an $MK\times1 $ vector as follows: 
\begin{align}
\mathbf{d}_p =  \begin{bmatrix} h_{\mathbf{S}_{p1} 1}^S \mathbf{1}_{1\times M}&\cdots &\xi h_{\mathbf{S}_{p1} \mathbf{S}_{p2} }^S\tilde{\mathbf{1}}_{1\times M}^{\mathbf{S}_{p1}} &\cdots  &h_{\mathbf{S}_{p1} K}^S \mathbf{1}_{1\times M} \end{bmatrix}^T,
\end{align} 
$\tilde{\mathbf{1}}_{1\times M}^{\mathbf{S}_{p1}} $ is a $1\times M$ all-one vector except that its $\mathbf{S}_{p1}$-th element is zero, 
$\mathbf{e}_p$ is an $MK\times1 $ vector as follows:
\begin{align}
\mathbf{e}_p = \begin{bmatrix} \frac{h_{\mathbf{S}_{p2} 1}^P}{h_{\mathbf{S}_{p2} \mathbf{S}_{p2} }^P} \mathbf{1}_{1\times M}&\cdots & \tilde{\mathbf{0}}^{\mathbf{S}_{p1} }_{1\times M} &\cdots &\frac{h_{\mathbf{S}_{p2} K}^P}{h_{\mathbf{S}_{p2} \mathbf{S}_{p2} }^P} \mathbf{1}_{1\times M}\end{bmatrix}^T,
\end{align} 
  $ \tilde{\mathbf{0}}^{\mathbf{S}_{p1} }_{1\times M}$ is a $1\times M$ all-zero vector except that its $\mathbf{S}_{p1}$-th element is one, and 
$\mathbf{f}_p$ is an $MK\times1 $ vector as follows:
\begin{align}
\mathbf{f}_p = \begin{bmatrix} \frac{h_{\mathbf{S}_{p1} 1}^S}{h_{\mathbf{S}_{p1} \mathbf{S}_{p2} }^S} \mathbf{1}_{1\times M}&\cdots & \tilde{\mathbf{0}}^{\mathbf{S}_{p1} }_{1\times M} &\cdots & \frac{h_{\mathbf{S}_{p1} K}^S}{h_{\mathbf{S}_{p1} \mathbf{S}_{p2} }^S} \mathbf{1}_{1\times M}\end{bmatrix}^T. 
\end{align} 
The concise expression shown in \eqref{pb:5} provides   insight to the considered joint beam and power allocation problem. For example,  the constraints in \eqref{5ttst:1} and \eqref{5ttst:3} are based on simple affine functions. However,  the constraint in \eqref{5ttst:2} is not in a convex form due to the involvement of   ${\rm sign}$. In addition, the objective function is also not in a concave form as its logarithm  terms contain ratios of linear functions.

\section{Optimal and Suboptimal Solutions for Joint Beam and Power Allocation}
In this section, two algorithms with different tradeoffs between performance and complexity are developed.

\subsection{Applying the Branch and Bound Method}
The BB method can be viewed as a type of structured search, where the feasibility region of the optimization problem is divided into smaller regions and a search for the optimal solution is carried out by focusing on those regions which are more promising and removing (pruning) those unlike ones. In the following,  the considered optimization problem is first reformulated    to facilitate the application of the BB method, and then the issue to find the upper and lower bounds on the optimal value of the considered optimization problem is focused. 

\subsubsection{Implementation  of the BB method} Recall that the BB method  can be ideally applied to the optimization problem with the following two features \cite{5765556,8170332}. One feature is that the feasibility region of the optimization problem can be expressed as   a multi-dimensional rectangle, as it can be straightforwardly partitioned. The other feature is that the lower and upper bounds on the objective function for a partitioned feasibility region can be straightforwardly found. 
In order to recast the considered optimization problem to a form with the aforementioned two features, by introducing auxiliary variables, $x_p$, $1\leq p\leq |\mathcal{S}|$, problem \ref{pb:5} can be first rewritten as follows:
 \begin{problem}\label{pb:6} 
  \begin{alignat}{2}
\underset{\mathbf{y},x_p}{\rm{min}}  &\quad    
 f(\mathbf{x})=- \sum_{p=1} ^{|\mathcal{S}|}   \log\left( 1+x_p
\right) \label{6tobj:1} \\
\rm{s.t.} & \quad    
x_p\leq\frac{ \mathbf{c}_p^T\mathbf{y}}{ \mathbf{d}_p^T\mathbf{R}\mathbf{y}  +t_{\mathbf{S}_{p1}\mathbf{S}_{p2}}} \label{6ttst:1}, 1\leq p\leq |\mathcal{S}|
\\ &\quad \eqref{5ttst:1}, \eqref{5ttst:2}, \eqref{5ttst:3}.
  \end{alignat}
\end{problem}
Problem \ref{pb:6} facilitates the application of the BB method since its feasibility region with respect to $x_p$, i.e., \eqref{6ttst:1}, is a simple multi-dimensional rectangle.  By defining $\mathbf{x}=\begin{bmatrix}x_1&\cdots &x_{|\mathcal{S}|} \end{bmatrix}$ and absorbing the variables, $\mathbf{y}$, into the feasibility region, 
problem \ref{pb:6} can be further rewritten as follows:
\begin{problem}\label{pb:7} 
  \begin{alignat}{2}
\underset{ \mathbf{x}}{\rm{min}}  &\quad    
 f(\mathbf{x})  \quad 
\rm{s.t.}  \quad    
\mathbf{x}\in \mathcal{G}. 
  \end{alignat}
\end{problem}
where 
\begin{align}
\mathcal{G}=\left\{\mathbf{x}|\eqref{6ttst:1}, \eqref{5ttst:1}, \eqref{5ttst:2}, \eqref{5ttst:3}\right\}.
\end{align}
As to be shown in the next subsection, the task to find the lower and upper on the optimal value can be simplified if  the objective function of the considered optimization problem is a monotonically decreasing  function, which motivates     the objective function to be rewritten as follows:
\begin{align}\label{new object}
 \tilde{f}(\mathbf{x})=\left\{\begin{array}{ll}- \sum_{p=1} ^{|\mathcal{S}|}   \log\left( 1+x_p
\right),& {\rm if} \quad \mathbf{x}\in \mathcal{G}\\
0, &{\rm otherwise}\end{array}\right..
\end{align}
It is straightforward to verify that $ \tilde{f}(\mathbf{x})$ is  a monotonically decreasing function of $\mathbf{x}$, and problem \ref{pb:7} can be expressed as follows:
\begin{problem}\label{pb:8} 
  \begin{alignat}{2}
\underset{ \mathbf{x}}{\rm{min}}  &\quad    
 \tilde{f}(\mathbf{x})  \quad 
\rm{s.t.}  \quad    
\mathbf{x}\in \mathcal{G}. 
  \end{alignat}
\end{problem}
As shown in \cite{8170332}, problem \ref{pb:8} yields   the same optimal solution as problem \ref{pb:7}. Problem \ref{pb:8} can be solved efficiently by applying the BB method, as shown  in Algorithm \ref{algorithm}.

 \begin{algorithm}[t]
\caption{Branch and Bound Algorithm}

 \begin{algorithmic}[1]
 
\State Set $\mathcal{B}_0=\{\mathcal{D}_0\}$ and   tolerance $\epsilon$,  $k=0$,  $U_0=\phi^{\rm up}(\mathcal{D}_0)$, $L_0=\phi^{\rm lb}(\mathcal{D}_0)$, and $\Delta=U_0-L_0$
\While { $\Delta\geq \epsilon$ }
\State $k=k+1$. 
\State  Find $\mathcal{D}\in\mathcal{B}_{k-1}$ which yields the smallest lower bound

\State Split $\mathcal{D}$ along its longest edge into $\mathcal{D}_{1}$ and $\mathcal{D}_2$

\State Construct $\mathcal{B}_k=\{\mathcal{D}_1\cup \mathcal{D}_1 \cup (\mathcal{B}_{k-1}\backslash \mathcal{D}) \}$

\State Find the updated upper bound $U_k=\min \phi^{\rm up}(\mathcal{D})$, $\forall \mathcal{D}\in \mathcal{B}_k$.  

\State Find the updated lower bound $L_k=\min \phi^{\rm lb}(\mathcal{D})$, $\forall \mathcal{D}\in \mathcal{B}_k$.  

\State  Update $\Delta$ by using $\Delta=U_k-L_k$

\State Remove those $\mathcal{D}$ in $\mathcal{B}_k$ whose lower bounds are larger than $U_k$.

 \EndWhile
\State \textbf{end}
 \end{algorithmic}\label{algorithm}
\end{algorithm}

As shown  in Algorithm \ref{algorithm}, the initialization of the BB method needs to define an initial search region, which is a $|\mathcal{S}|$-dimensional rectangle, denoted by $\mathcal{D}_0$ and defined as follows:
\begin{align}
\mathcal{D}_0 = \{\mathbf{x}|0\leq x_p\leq x_{p}^B\},
\end{align}
where $x_{p}^B = \frac{ P_{\max}\mathbf{c}_p^T\mathbf{1}_{|\mathcal{S}|\times 1}  }{  t_{\mathbf{S}_{p1}\mathbf{S}_{p2}}}$ denotes  the upper bound on $x_p$. 

During the $k$-th iteration, the rectangle which yields the smallest lower bound among all the   rectangles  in the set, $\mathcal{B}_{k-1}$,   is located and partitioned into two smaller rectangles along its longest edge, denoted by $\mathcal{D}_1$ and $\mathcal{D}_2$, respectively. For each rectangle, $\mathcal{D}_i$, find   the upper and lower bounds on the optimal value of the considered optimization problem, denoted by $\phi^{\rm up}(\mathcal{D}_i)$ and $\phi^{\rm lb}(\mathcal{D}_i)$, respectively.  More discussions will be provided in the next subsection for calculating $\phi^{\rm up}(\mathcal{D})$ and $\phi^{\rm lb}(\mathcal{D})$. The overall lower and upper bounds on the optimal value can be updated iteratively   as shown in Algorithm \ref{algorithm}, where this updated  upper bound can also be used to remove those rectangles whose lower bounds are larger than the new upper bound. The algorithm terminates when the difference between the overall lower and upper bounds is smaller than the given tolerance parameter, denoted by $\epsilon$. 

\subsubsection{Finding the upper and lower bounds $\phi^{\rm up}(\mathcal{D})$ and $\phi^{\rm lb}(\mathcal{D})$}

For each rectangle, $\mathcal{D}$, denote its maximum and minimum vertices by $\mathbf{x}_{\max}$ and $\mathbf{x}_{\min}$, respectively, i.e., $\mathbf{x}_{\min}\leq \mathbf{x}\leq  \mathbf{x}_{\max}$ for $\mathbf{x}\in \mathcal{D}$. By using the fact that $\tilde{f}(\mathbf{x})$ is a monotonically decreasing function of $\mathbf{x}$, the   lower and upper bounds,   $\phi^{\rm lb}(\mathcal{D})$ and $\phi^{\rm up}(\mathcal{D})$, can be calculated as follows:
\begin{align}\label{low bound1}
\phi^{\rm lb}(\mathcal{D}) = \left\{\begin{array}{ll}
f(\mathbf{x}_{\max}), &{\rm if }\quad \mathbf{x}_{\min}\in \mathcal{G}\\
0, &{\rm otherwise }
\end{array}\right.,
\end{align}
and
\begin{align}\label{upper bound1}
\phi^{\rm up}(\mathcal{D}) = \left\{\begin{array}{ll}
f(\mathbf{x}_{\min}), &{\rm if }\quad \mathbf{x}_{\min}\in \mathcal{G}\\
0, &{\rm otherwise }
\end{array}\right..
\end{align}
The rationale behind the  bounds shown in \eqref{low bound1} and \eqref{upper bound1} is that, if $\mathbf{x}_{\min}$ is feasible,  $\mathbf{x}_{\min}$   yields the maximal value for the objective function since it is the minimum element in $\mathcal{D}$, whereas $\mathbf{x}_{\max}$ is the maximum element in $\mathcal{D}$ and can be used to find a (not necessarily achievable) lower bound on the objective function. If    $\mathbf{x}_{\min}$ is not feasible, the rectangle is located outside of the feasibility region and hence should be pruned (removed), where the use of zero for the upper and lower bounds in this case can realize this goal. The task to verify whether $\mathbf{x}_{\min}\in \mathcal{G}$, i.e., $\mathbf{x}_{\min}$ is feasible, can be accomplished by carrying out the following feasibility study
  \begin{problem}\label{pb:9} 
  \begin{alignat}{2}
\underset{\mathbf{y}}{\rm{max}}  &\quad    1\label{9tobj:1} \\
\rm{s.t.} & \quad     \left( \mathbf{c}_p^T-x_p \mathbf{d}_p^T\mathbf{R}\right)\mathbf{y}\geq   x_{\min, p} t_{\mathbf{S}_{p1}\mathbf{S}_{p2}}\label{9ttst:1}
  \\&\quad 
\mathbf{e}_p^T\mathbf{R}\mathbf{y}+c_{\mathbf{S}_{p2}}
 \leq 0, 1\leq p\leq |\mathcal{S}|\label{9ttst:2}
\\ &\quad  {\rm sign}(y_p) \left(\mathbf{f}_p ^T\mathbf{R}\mathbf{y}  +b_{\mathbf{S}_{p1}\mathbf{S}_{p2}}\right)\leq 0 , 1\leq p\leq |\mathcal{S}| \label{9ttst:3}
\\
&\quad \mathbf{1}_{|\mathcal{S}|\times 1}^T\mathbf{y}\leq P^{\rm max}\label{9ttst:4},
  \end{alignat}
\end{problem}
where $x_{\min,p}$ is the $p$-th element of $\mathbf{x}_{\min}$. 
Note that the use of ${\rm sign}(y_p) $ makes constraint \eqref{9ttst:3} not convex. It is important to point out that $y_p=0$ and $x_{\min, p}=0$ are equivalent, i.e.,  $y_p=0$ leads to $x_{\min, p}=0$ and vice versa. By using this observation,  problem \ref{pb:9} can be recast as the following equivalent  form:
 \begin{problem}\label{pb:10} 
  \begin{alignat}{2}
\underset{\mathbf{y}}{\rm{max}}  &\quad    1\label{10tobj:1} \\
\rm{s.t.} & \quad   {\rm sign}(x_{\min,p}) \left(\mathbf{f}_p ^T\mathbf{R}\mathbf{y}  +b_{\mathbf{S}_{p1}\mathbf{S}_{p2}}\right)\leq 0 , 1\leq p\leq |\mathcal{S}| \label{10ttst:1}
\\
&\quad \eqref{9ttst:1},\eqref{9ttst:2},\eqref{9ttst:4}.
  \end{alignat}
\end{problem}
Note that in problem \ref{pb:10}, $\mathbf{x}_{\min}$ is not an optimization variable and hence   ${\rm sign}(x_{\min, p})$ is  a constant, which means that constraint \eqref{10ttst:1} is an affine function. Since all the constraint functions  of problem \ref{pb:10} are affine, problem \ref{pb:10} can be solved efficiently by applying those off-shelf optimization solvers. 
 
Because the convergency of the BB method depends on how tight the upper and lower bounds are. Therefore, the bounds shown in  \eqref{low bound1} and \eqref{upper bound1} will be further tightened, by using the steps shown in \cite{5765556}.  
  
\subsubsection{Tightening the upper and lower bounds} In this section, we will focus on the case that $\mathbf{x}_{\min}\in\mathcal{G}$, otherwise the corresponding rectangle will be eventually  pruned. Furthermore, we assume that $\mathbf{x}_{\max}\notin\mathcal{G}$, otherwise the bounds shown in \eqref{low bound1} and \eqref{upper bound1}  are tight.  In  \eqref{low bound1}, a lower bound is obtained by directly using $\mathbf{x}_{\max}$ which is often   far away from the boundary of the feasibility region. The key idea for getting a lower bound tighter than \eqref{low bound1} is to find a new vector, denoted by $\tilde{\mathbf{x}}_{\max}$, which is closer to the boundary of the feasibility region than  $\mathbf{x}_{\max}$.  In particular, the $p$-th element of  $\tilde{\mathbf{x}}_{\max}$, denoted by $\tilde{ {x}}_{\max,p}$, is obtained by finding its maximal value if all the other elements of $\tilde{\mathbf{x}}_{\max}$ are the same as those of $ {\mathbf{x}}_{\min}$, which is to solve the following optimization problem  
 \begin{problem}\label{pb:11} 
  \begin{alignat}{2}
\underset{\mathbf{y} }{\rm{max}}  &\quad    
\tilde{ {x}}_{\max,p}\triangleq  \frac{ \mathbf{c}_p^T\mathbf{y}}{ \mathbf{d}_p^T\mathbf{R}\mathbf{y}  +t_{\mathbf{S}_{p1}\mathbf{S}_{p2}}} \label{11tobj:1} \\
\rm{s.t.} & \quad    \frac{ \mathbf{c}_p^T\mathbf{y}}{ \mathbf{d}_p^T\mathbf{R}\mathbf{y}  +t_{\mathbf{S}_{p1}\mathbf{S}_{p2}}}\leq  x_{\max, p}\\&\quad 
x_{\min, i}= \frac{ \mathbf{c}_i^T\mathbf{y}}{ \mathbf{d}_i^T\mathbf{R}\mathbf{y}  +t_{\mathbf{S}_{i1}\mathbf{S}_{i2}}}, i\neq p,1\leq i\leq |\mathcal{S}|\\&\quad 
 \eqref{9ttst:2}, \eqref{10ttst:1},\eqref{9ttst:4}.
  \end{alignat}
\end{problem}
Problem \ref{pb:11} can be rewritten as follows:
 \begin{problem}\label{pb:12} 
  \begin{alignat}{2}
\underset{\mathbf{y} }{\rm{max}}  &\quad    
 \frac{ \mathbf{c}_p^T\mathbf{y}}{ \mathbf{d}_p^T\mathbf{R}\mathbf{y}  +t_{\mathbf{S}_{p1}\mathbf{S}_{p2}}} \label{12tobj:1} \\
\rm{s.t.} & \quad    \left( \mathbf{c}_p^T-x_{\max,p}  \mathbf{d}_p^T\mathbf{R}\right)\mathbf{y}  \leq   x_{\max,p}  t_{\mathbf{S}_{p1}\mathbf{S}_{p2}}\\&\quad 
\left(\mathbf{c}_i^T-x_{\min,i} \mathbf{d}_i^T\mathbf{R}\right)\mathbf{y}= x_{\min,i}  t_{\mathbf{S}_{i1}\mathbf{S}_{i2}}, i\neq p, 1\leq i\leq |\mathcal{S}|\\ &\quad 
\mathbf{e}_p^T\mathbf{R}\mathbf{y}\leq -c_{\mathbf{S}_{p2}}
 , 1\leq p\leq |\mathcal{S}|\label{12ttst:1}
\\ &\quad x_{\min, p}   \mathbf{f}_p ^T\mathbf{R}\mathbf{y}  \leq -x_{\min, p}  b_{\mathbf{S}_{p1}\mathbf{S}_{p2}} , 1\leq p\leq |\mathcal{S}|\label{12ttst:2}
\\
&\quad \mathbf{1}_{|\mathcal{S}|\times 1}^T\mathbf{y}\leq P^{\rm max}.\label{12ttst:3}
  \end{alignat}
\end{problem}
Note that all the constraints of problem \ref{pb:12} are affine, and hence they can be grouped to yield the following more concise form:
 \begin{problem}\label{pb:13} 
  \begin{alignat}{2}
\underset{\mathbf{y} }{\rm{max}}  &\quad    
 \frac{ \mathbf{c}_p^T\mathbf{y}}{ \mathbf{d}_p^T\mathbf{R}\mathbf{y}  +t_{\mathbf{S}_{p1}\mathbf{S}_{p2}}} \label{13tobj:1} \\
\rm{s.t.} & \quad    \mathbf{A}_p \mathbf{y}\leq \mathbf{b}_p  \label{13ttst:1} 
\\&\quad \mathbf{A}_p^E \mathbf{y}= \mathbf{b}_p^E \label{13ttst:2} , \end{alignat}
\end{problem}
where the expressions for $\mathbf{A}_p$, $\mathbf{b}_p$, $\mathbf{A}_p^E$ and $\mathbf{b}_p^E$ can be straightforwardly obtained from problem \ref{pb:12} and are omitted due to space limitations. 

Define   $\tilde{\mathbf{A}}_p^E$ as a $(|\mathcal{S}|-1)\times (|\mathcal{S}|-1)$ square matrix obtained from $\mathbf{A}_p^E$ by removing its $p$-th column of $ {\mathbf{A}}_p^E$, denoted by   $\mathbf{a}_{p,p}^E$.  $\tilde{\mathbf{y}}_p$ is a  $(|\mathcal{S}|-1)\times 1$ vector obtained from $ {\mathbf{y}}$ by removing $y_p$. A close-form expression for the optimal solution of problem \ref{pb:13} can be obtained as follows.

\begin{lemma}\label{lemma1}
Assume that $\tilde{\mathbf{A}}_p^E$ is invertible and problem \ref{pb:13} is feasible. 
The optimal solution for problem \ref{pb:13}, denoted by  $ {\mathbf{y}}_p^*$, can be obtained as follows. First, the $p$-th element of  $ {\mathbf{y}}_p^*$ can be expressed as follows:
\begin{align}
y_p^*=\min\{\mathbf{a}_{\rm sign}\odot(\mathbf{b}_p -\tilde{\mathbf{A}}_p(\tilde{\mathbf{A}}_p^E)^{-1} \mathbf{b}_p^E)./ (\mathbf{a}_{p,p}-\tilde{\mathbf{A}}_p(\tilde{\mathbf{A}}_p^E)^{-1} \mathbf{a}_{p,p}^E )\},
\end{align}
where  $\mathbf{a}_{p,p}$ denotes the $p$-th column of $ \mathbf{A}_p$, $\tilde{\mathbf{A}}_p$ is obtained from $ \mathbf{A}_p$ by removing  $\mathbf{a}_{p,p}$, $\mathbf{a}_{\rm sign}=\max\left\{0,(\mathbf{a}_{p,p}-\tilde{\mathbf{A}}_p(\tilde{\mathbf{A}}_p^E)^{-1} \mathbf{a}_{p,p}^E )\right\}$,  $\odot$ and $./$ denote element-wise multiplication and division, respectively. Second, 
collect the remaining $(|\mathcal{S}|-1)$ elements of $ {\mathbf{y}}_p^*$ in the vector, denoted by    $  \tilde{\mathbf{y}}_p^*$, and $ \tilde{\mathbf{y}}_p^*$ can be obtained from $y_p^*$ as follows: $ \tilde{\mathbf{y}}_p^*=(\tilde{\mathbf{A}}_p^E)^{-1} (\mathbf{b}_p^E-\mathbf{a}_{p,p}^E  {y}_p^*)$. 
\end{lemma}
\begin{proof}
See Appendix \ref{proof1}
\end{proof}

Among our conducted computer simulations, we notice that    $\tilde{\mathbf{A}}_p^E$ can be close to singular for a small number of  channel   realizations. For these rare cases, problem \ref{pb:13} can still be solved efficiently by first reformulating it as a linear programming problem and then applying those optimization solvers. In particular,   define $\mathbf{z}=\frac{  \mathbf{y}}{ \mathbf{d}_p^T\mathbf{R}\mathbf{y}  +t_{\mathbf{S}_{p1}\mathbf{S}_{p2}}} $ and $w=\frac{1}{ \mathbf{d}_p^T\mathbf{R}\mathbf{y}  +t_{\mathbf{S}_{p1}\mathbf{S}_{p2}}} $, which means that problem \ref{pb:13} can be recast as follows:
 \begin{problem}\label{pb:14} 
  \begin{alignat}{2}
\underset{\mathbf{y} }{\rm{max}}  &\quad    
 \mathbf{c}_p^T\mathbf{z}  \label{14tobj:1} \\
\rm{s.t.} & \quad    \mathbf{A}_p \mathbf{z}\leq w\mathbf{b}_p  \label{14ttst:1} 
\\&\quad \mathbf{A}_p^E \mathbf{z}= w\mathbf{b}_p^E \label{14ttst:2}
\\
&\quad \mathbf{d}_p^T\mathbf{R}\mathbf{z}+t_{\mathbf{S}_{p1}\mathbf{S}_{p2}} w=1 \label{14ttst:3}.
  \end{alignat}
\end{problem}
Problem \ref{pb:14} is in a linear programming form and hence can be directly solved by applying those off-shelf optimization solvers.  Once all the elements of $\tilde{\mathbf{x}}_{\max}$ are obtained, a tighter lower bound can be found by  replacing $ {\mathbf{x}}_{\max}$ with $\tilde{\mathbf{x}}_{\max}$ in \eqref{low bound1}. 

Interestingly, $\tilde{\mathbf{x}}_{\max}$ can also be used to tighten the upper bound.  Recall that the $p$-th element of $\tilde{\mathbf{x}}_{\max}$ is obtained by first assuming that the other elements of $\tilde{\mathbf{x}}_{\max}$ are equal to those in $ {\mathbf{x}}_{\min}$ and then solving problem \ref{pb:13}. Therefore, build the $|\mathcal{S}|$ vectors, denoted by $\tilde{\mathbf{x}}_{\min}^i$, $1\leq i \leq |\mathcal{S}|$, where each of the vectors is a $|\mathcal{S}|\times 1$ vector, and its $p$-th element, denoted by $\tilde{ {x}}_{\min,p}^i$, is given by 
\begin{align}
\left\{
\begin{array}{ll}
\tilde{ {x}}_{\min,p}^i = \tilde{ {x}}_{\max,p}&\text{if } p=i
\\
\tilde{ {x}}_{\min,p}^i =  { {x}}_{\min,p}&\text{if } p\neq i
\end{array}
\right..
\end{align}
  According to the steps to find $ \tilde{ {x}}_{\max,p}$, all the vectors, $\tilde{\mathbf{x}}_{\min}^i$, $1\leq i \leq |\mathcal{S}|$, are feasible, i.e., $\tilde{\mathbf{x}}_{\min}^i\in \mathcal{G}$. Therefore, these vectors can be used to form new upper bounds on the optimal value. In particular, by replacing   $ {\mathbf{x}}_{\min}^i$ with $\tilde{\mathbf{x}}_{\min}^i$, $1\leq i \leq |\mathcal{S}|$, in \eqref{upper bound1},  $|\mathcal{S}|$ new upper bounds can be obtained, where     the smallest one  can be used as the tightened upper bound. 

\subsection{Applying the Successive Convex Approximation Method}
Recall that the BB method is essentially a structured search, where many iterations are required in order to  divide the feasibility region into sufficiently small multi-dimensional rectangles. As a result, the computational complexity of the BB method can be significant, particularly in the case that the number of optimization variables, $|\mathcal{S}|$, is large, which motivates the use of the SCA method. 

In order to facilitate the application of SCA, problem \ref{pb:5} can be first  re-written as follows:
 \begin{problem}\label{pb:15} 
  \begin{alignat}{2}
\underset{\mathbf{y}}{\rm{max}}  &\quad    
  \sum_{p=1} ^{|\mathcal{S}|}  \left[ \log\left(   \mathbf{c}_p^T\mathbf{y}+ \mathbf{d}_p^T\mathbf{R}\mathbf{y}  +t_{\mathbf{S}_{p1}\mathbf{S}_{p2}}
\right)- \log\left(   \mathbf{d}_p^T\mathbf{R}\mathbf{y}  +t_{\mathbf{S}_{p1}\mathbf{S}_{p2}}
\right) \right] \label{15tobj:1} \\
\rm{s.t.} & \quad    
\mathbf{e}_p^T\mathbf{R}\mathbf{y}+c_{\mathbf{S}_{p2}}
 \leq 0, 1\leq p\leq |\mathcal{S}|\label{15ttst:1}
\\ &\quad  {\rm sign}(y_p) \left(\mathbf{f}_p ^T\mathbf{R}\mathbf{y}  +b_{\mathbf{S}_{p1}\mathbf{S}_{p2}}\right)\leq 0 , 1\leq p\leq |\mathcal{S}| \label{15ttst:2}
\\
&\quad \mathbf{1}_{|\mathcal{S}|\times 1}^T\mathbf{y}\leq P^{\rm max}.\label{15ttst:3}
  \end{alignat}
\end{problem}
To tackle the challenge that the objective function of problem \ref{pb:15} is not concave, auxiliary optimization variables, $z_p$, are introduced and problem \ref{pb:15} can be equivalently recast as follows:
 \begin{problem}\label{pb:16} 
  \begin{alignat}{2}
\underset{\mathbf{y},z_p}{\rm{max}}  &\quad    
  \sum_{p=1} ^{|\mathcal{S}|}   \log\left(   \mathbf{c}_p^T\mathbf{y}+ \mathbf{d}_p^T\mathbf{R}\mathbf{y}  +t_{\mathbf{S}_{p1}\mathbf{S}_{p2}}
\right)- \sum_{p=1} ^{|\mathcal{S}|}  z_p \label{16tobj:1} \\
\rm{s.t.} & \quad    \log\left(   \mathbf{d}_p^T\mathbf{R}\mathbf{y}  +t_{\mathbf{S}_{p1}\mathbf{S}_{p2}}
\right)\leq z_p, 1\leq p\leq |\mathcal{S}| \label{16ttst:1}
\\ &\quad  \eqref{15ttst:1}, \eqref{15ttst:2}, \eqref{15ttst:3}.
  \end{alignat}
\end{problem}
Note that  the constraint in \eqref{16ttst:1} is not convex, but it can be approximated by using the first order Taylor expansion, which means that problem \ref{pb:15} can be approximated  as follows: 
  \begin{problem}\label{pb:17} 
  \begin{alignat}{2}
\underset{\mathbf{y},z_p}{\rm{max}}  &\quad   \sum_{p=1} ^{|\mathcal{S}|}   \log\left(   \mathbf{c}_p^T\mathbf{y}+ \mathbf{d}_p^T\mathbf{R}\mathbf{y}  +t_{\mathbf{S}_{p1}\mathbf{S}_{p2}}
\right)- \sum_{p=1} ^{|\mathcal{S}|}  z_p \label{16tobj:1} \\
\rm{s.t.} & \quad   \log\left(  \mathbf{d}_p^T\mathbf{R}\mathbf{y}_0  +t_{\mathbf{S}_{p1}\mathbf{S}_{p2}}
\right) +   \frac{\mathbf{d}_p^T\mathbf{R}\left(\mathbf{y}-\mathbf{y}_0\right)}{\ln (2)\left(  \mathbf{d}_p^T\mathbf{R}\mathbf{y}_0  +t_{\mathbf{S}_{p1}\mathbf{S}_{p2}}\right)
}\leq z_p, \forall p  \label{17ttst:1}\\&
\quad  \eqref{15ttst:1}, \eqref{15ttst:2}, \eqref{15ttst:3}.
  \end{alignat}
\end{problem}
where $\mathbf{y}_0$ denotes an initial  estimate of $\mathbf{y}$  and can be iteratively updated.  
It is straightforward to verify that the objective function of problem \ref{pb:17} is concave, and the newly introduced constraint in \eqref{17ttst:1} is a simple affine function. 

The only remaining challenge to solve problem \ref{pb:17} is that constraint  \eqref{15ttst:2} is still not in a convex form due to the use of the sign function. In the following, two heuristic solutions, termed SCA-I and SCA-II, respectively, are proposed to reformulate problem \ref{pb:17} into a concave optimization form. SCA-I is to directly remove the sign function in \eqref{15ttst:2}, which leads to the following optimization problem:
 \begin{problem}\label{pb:18} 
  \begin{alignat}{2}
\underset{\mathbf{y},z_p}{\rm{max}}  &\quad   \sum_{p=1} ^{|\mathcal{S}|}   \log\left(   \mathbf{c}_p^T\mathbf{y}+ \mathbf{d}_p^T\mathbf{R}\mathbf{y}  +t_{\mathbf{S}_{p1}\mathbf{S}_{p2}}
\right)- \sum_{p=1} ^{|\mathcal{S}|}  z_p \label{16tobj:1} \\
\rm{s.t.} & \quad       \left(\mathbf{f}_p ^T\mathbf{R}\mathbf{y}  +b_{\mathbf{S}_{p1}\mathbf{S}_{p2}}\right)\leq 0 , 1\leq p\leq |\mathcal{S}|  \label{18ttst:1}\\&
\quad  \eqref{15ttst:1}, \eqref{15ttst:3}, \eqref{17ttst:1},
  \end{alignat}
\end{problem}
which is a typical concave maximization problem, and can be solved efficiently by applying the optimization solvers.  

SCA-II is motivated by the fact that the challenge in constraint \eqref{15ttst:2} is caused by  the use of the beam assignment indicator function. If   beam assignment is carried out before power allocation, this challenging issue can be avoided. Therefore, SCA-II  consists of two steps. The first step is to carry out user scheduling on each beam,  i.e., secondary user ${\rm U}_{j^*_k}^k$ is scheduled on beam $\mathbf{f}_k$ if $j^*_k = \underset{j}{\arg} \max \{h^S_{j,k}, \{j,k\}\in \mathcal{S}\}$.    The second step of SCA-II  is to update $S$ by including $\{j^*_k,k\}$, $1\leq k \leq K$, only, and then carry out power allocation, i.e., solving problem \ref{pb:18} in the same manner as    SCA-I.

 \section{Simulation Results}
 In this section, the computer simulation results are presented to evaluate the performance of THz-NOMA with joint beam and power allocation. Motivated by Lemma \ref{lemma0},    the greedy scheduling scheme is used as a benchmarking scheme for the BB and SCA methods. For all conducted simulations, $\rho_P=30$ dBm, $\sigma^2=-90$ dBm, $P_{\max}=30$ dBm, $\xi=10^8$,    $\alpha_{\rm PL}=2$, 
 $f_c=300$ GHz,  $\zeta=5e^{-3}$, $d=\frac{c}{2f_c}$, and  $\phi=0.1$,  as in \cite{jeffthz}. The primary users are   uniformly  deployed  within a square with its edge length $10$ m, where $\theta^P_k=\frac{i}{K}\pi-\frac{\pi}{2}$, $1\leq i\leq K$. The secondary users are also uniformly deployed within a square with its edge length denoted by $r_S$, where $\theta^S_j$ is uniformly distributed between $-\frac{\pi}{2}$ and $\frac{\pi}{2}$.  Recall that   the BB method can be viewed as a structured exhaustive search, and   its convergence  requires  a significant number of iterations, particularly if there are a large number of users.  A useful observation to reduce the complexity of the implementation of the BB method is that when sufficient iterations are carried out, the rectangles in $\mathcal{B}_k$ are already small enough to provide a good estimate for the optimal value.  Table I shows the effect of capping the number of iterations for the BB method, where $N_{itr}$ denotes the maximal  number of  iterations,    $N=10$, $K=4$, and $N_Q=10$. As can be seen from the table, capping the number of iterations does not cause a significant performance loss for the BB method, and hence $N_{itr}=200$ is used in the following conducted simulations.

\begin{table}
  \centering
  \caption{Impact of $N_{itr}$ on the Performance of the BB Method}
  \begin{tabular}{|c|c|c|c|c|c|}
\hline
&  $M=1$&$M=2$&$M=4$&$M=6$&$M=8$ \\
    \hline
   $N_{itr}=\infty$  &  $2.2805 $	&$4.04997$ &$5.7922$ &$6.9129$ &$7.8640$ \\
    \hline
   $N_{itr}=200$ & $ 2.2791$ &$3.8855$ &$5.7205$ &	$6.8343$  	&$7.4128$  \\\hline
  \end{tabular}\vspace{1em}\label{s}\vspace{-2em}
\end{table}

 \begin{figure}[t] \vspace{-0em}
\begin{center}
\subfigure[$\bar{R}_k=2.5$ BPCU]{\label{fig1a}\includegraphics[width=0.45\textwidth]{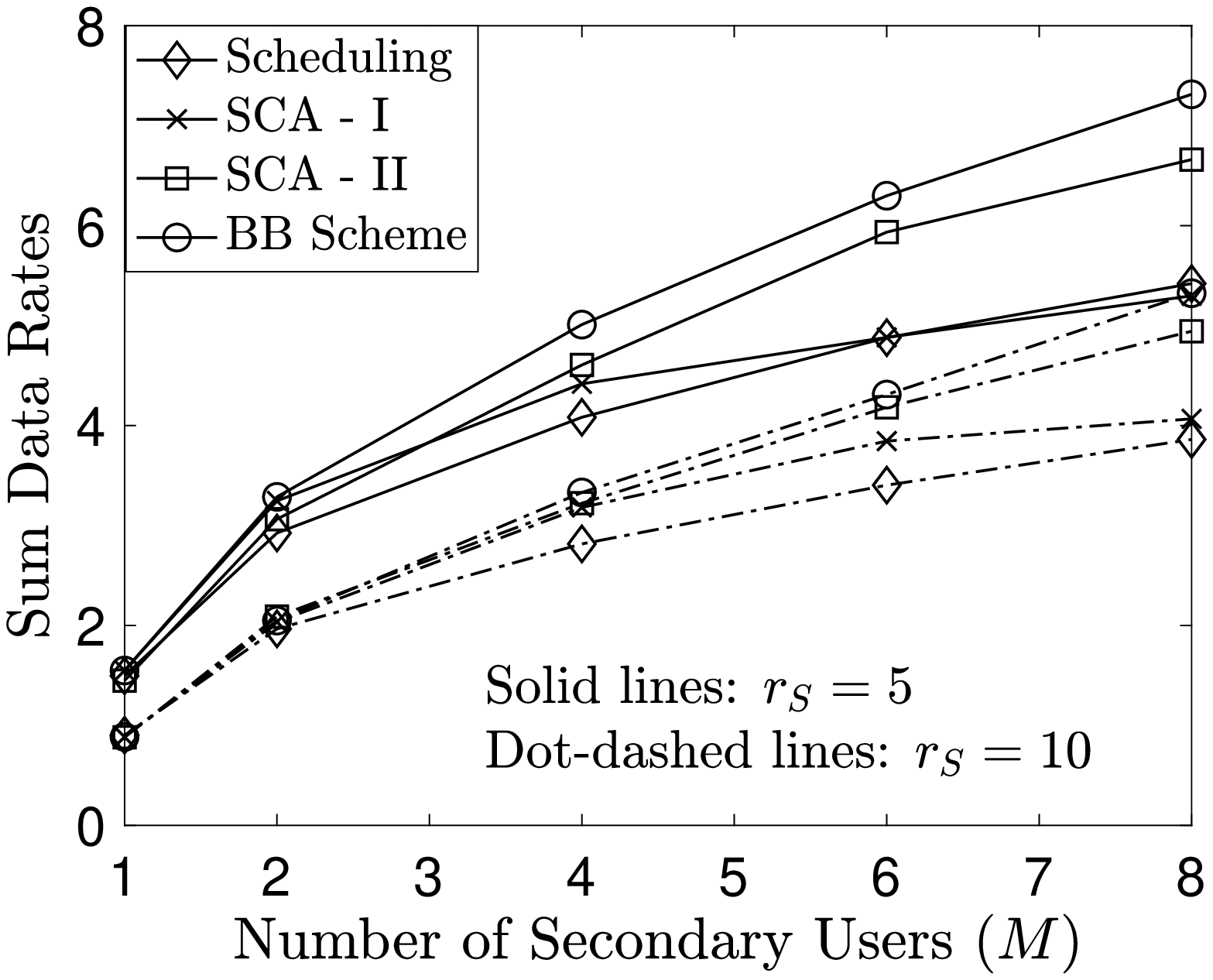}}\hspace{2em}
\subfigure[$\bar{R}_k=1$ BPCU]{\label{fig1b}\includegraphics[width=0.45\textwidth]{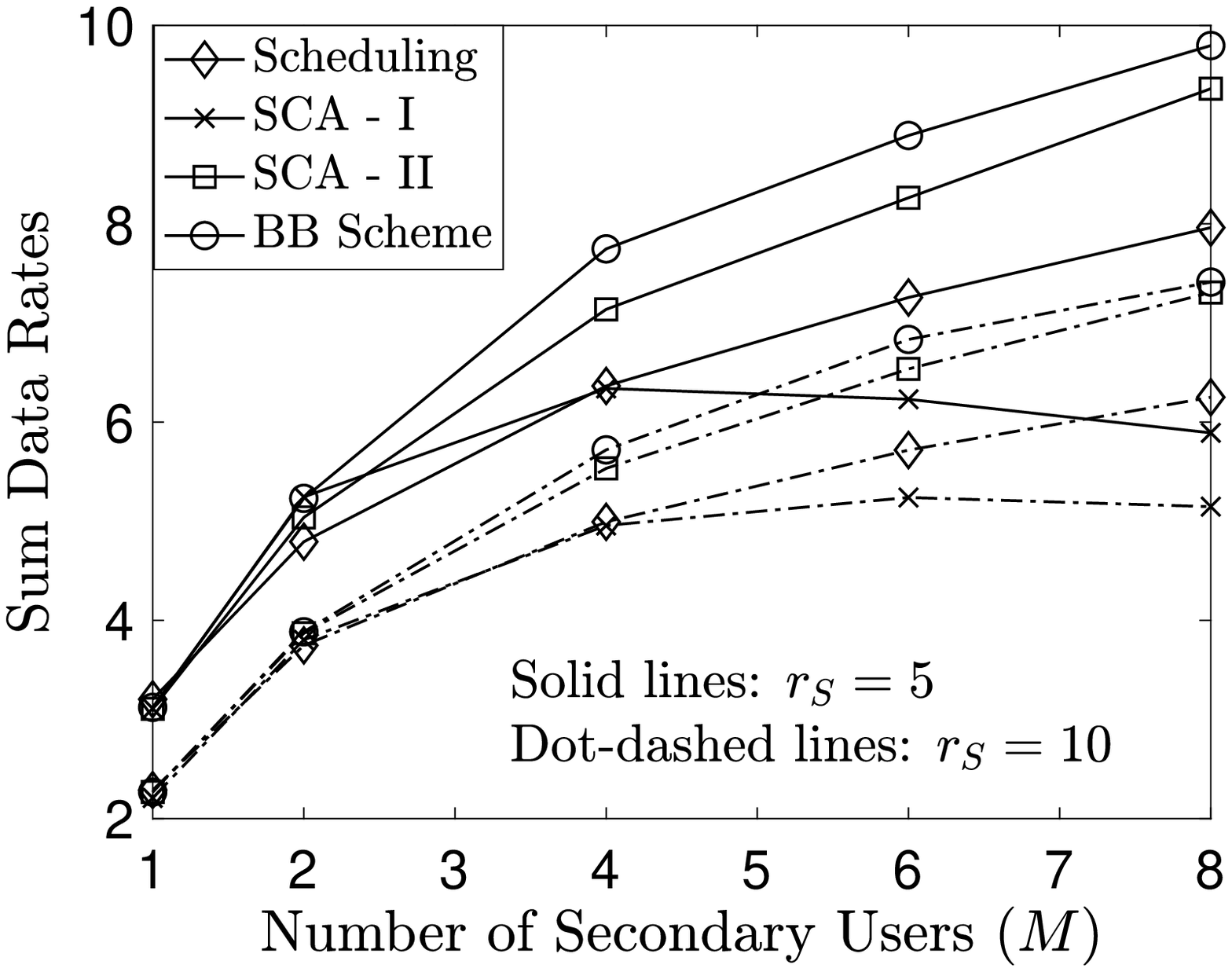}} \vspace{-1.5em}
\end{center}
\caption{ Impact of the number of secondary users on the performance of THz-NOMA.   $\alpha=2$,   $N=10$, $K=4$, $N_Q=10$. Hybrid beamforming is used.      \vspace{-1em} }\label{fig1}\vspace{-1.5em}
\end{figure}

 In Fig. \ref{fig1}, the impact of the number of secondary users on the performance of THz-NOMA is studied, where different choices of $r_S$ are used.  As can be seen from Fig. \ref{fig1}, the use of THz-NOMA can ensure that the secondary users are served on those existing beams with   significant data rates. This means that  the overall system throughput of THz networks can be significantly improved   compared to the case in which only the primary users are served.  Fig. \ref{fig1} also shows  that  the sum rate gain can be further increased by increasing $M$, or reducing $r_S$ and $\bar{R}_k$. Among the considered schemes, the BB method yields the best performance, since it is a structured search and is expected to provide the optimal performance. When there is a single secondary user,   the greedy scheduling scheme realizes the same performance as the BB method, which   confirms Lemma \ref{lemma0}. SCA-I is a naive application of the SCA method, and its performance can be even worse than the greedy scheduling scheme, particularly for the case with a larger $M$. SCA-II is the combination of user scheduling and SCA, and Fig. \ref{fig1} shows that the SCA-II can outperform the greedy scheduling scheme, and realize a performance close to the optimal BB method. It is important to point out that the convergence of SCA is much faster than the BB method, as shown in Fig. \ref{fig2}, which means that the complexity of SCA is much smaller than that of the BB method. In particular, Fig. \ref{fig2} demonstrates that SCA can converge within a single iteration, whereas the BB method can take hundreds of iterations to converge, even for the case with a moderate number of  users. 
    \begin{figure}[t]\centering \vspace{-2em}
    \epsfig{file=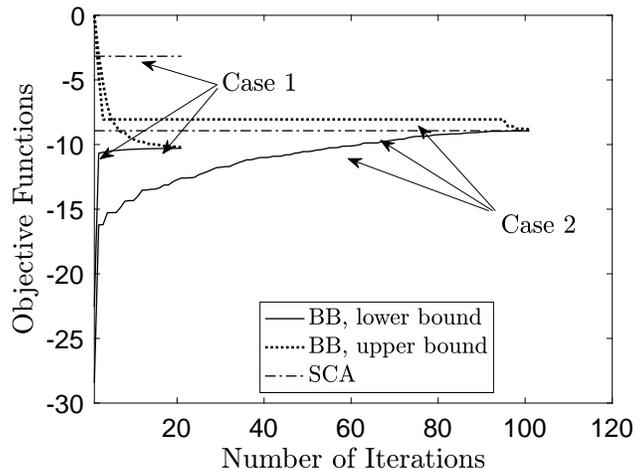, width=0.5\textwidth, clip=}\vspace{-0.5em}
\caption{ Illustration of the convergence of the SCA-II and BB schemes.  Two random realizations of channels are used. $N=10$, $K=4$, $M=8$, $R_k=2.5$ BPCU, $N_Q=10$ and $r_S=5$.
  \vspace{-1em}    }\label{fig2}   \vspace{-0.5em} 
\end{figure}

    \begin{figure}[t]\centering \vspace{-0em}
    \epsfig{file=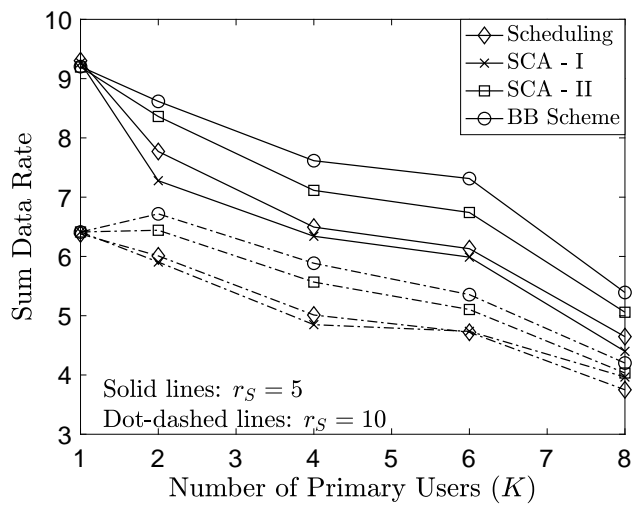, width=0.5\textwidth, clip=}\vspace{-0.5em}
\caption{ Impact of the number of primary users on the performance of THz-NOMA transmission.   $N=10$, $M=4$, $R_k=1$ BPCU, $N_Q=10$ and $r_S=10$.
  \vspace{-1em}    }\label{fig3}   \vspace{-2em} 
\end{figure}

In Fig. \ref{fig3}, the impact of the number of primary users on the performance of THz-NOMA transmission is studied. Recall that Fig. \ref{fig1} shows that inviting more secondary users to participate in THz-NOMA transmission can increase the overall sum rate, because a larger secondary user pool is helpful to improve    the effective channel gains of the scheduled secondary users. Intuitively, increasing $K$ should also be helpful to increase the sum rate, since there are more beams, i.e., there are more  bandwidth resources available. Fig. \ref{fig3} shows a surprising result that the performance of THz-NOMA is reduced when there are more primary users, which can be explained as follows. Unlike OFDMA subcarriers, the $K$ spatial  beams are not orthogonal bandwidth resources for the secondary users. In particular,   these beams have been    tailored to the primary users' channels in order to ensure that there is no inter-beam interference between the primary users. Because the secondary users' channels are different from  the primary users', the secondary users still experience inter-beam interference. This inter-beam interference can cause two types of performance degradation. First, each secondary user can suffer more interference from the primary users, if $K$ increases. Second, by increasing   $K$, more secondary users are scheduled, which further increases  interference in the network. 

    \begin{figure}[t]\centering \vspace{-2em}
    \epsfig{file=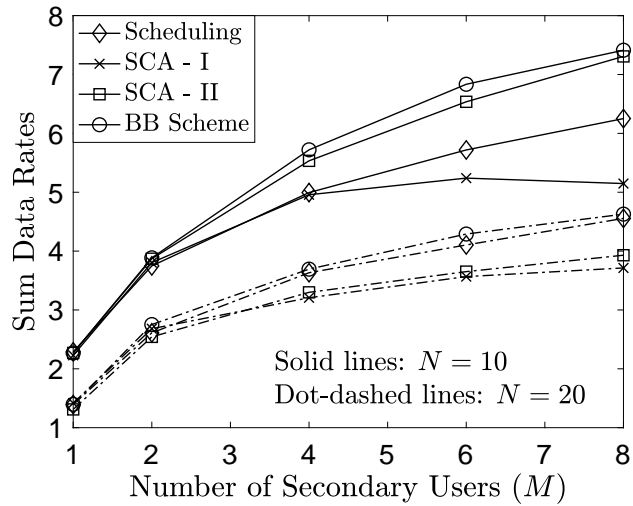, width=0.5\textwidth, clip=}\vspace{-0.5em}
\caption{ Impact of the number of antennas at the base station on the performance of THz-NOMA transmission.      $K=4$, $R_k=1$ BPCU, $N_Q=10$ and $r_S=10$.
  \vspace{-1em}    }\label{fig4}   \vspace{-1em} 
\end{figure}

    \begin{figure}[t]\centering \vspace{-0em}
    \epsfig{file=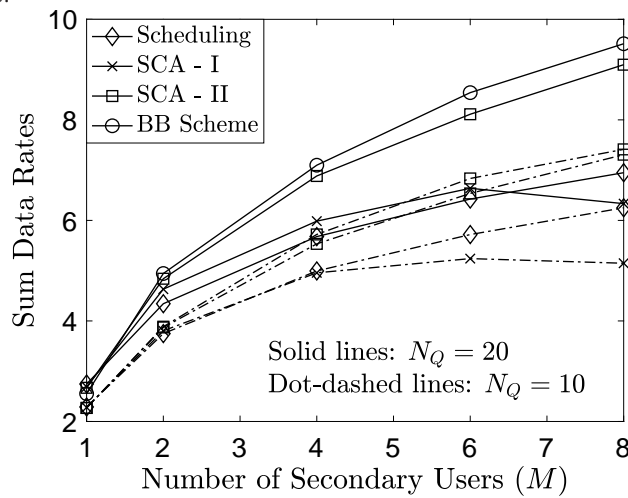, width=0.5\textwidth, clip=}\vspace{-0.5em}
\caption{ Impact of the   size of the beamsteering codebook   on the performance of THz-NOMA transmission.   $N=10$, $K=4$,$R_k=1$ BPCU,  and $r_S=10$.
  \vspace{-1em}    }\label{fig5}   \vspace{-2em} 
\end{figure}

In Fig. \ref{fig4}, the impact of the number of antennas at the base station on the performance of THz-NOMA networks is studied. As can be seen from the figure,  increasing the number of antennas at the base station reduces the sum rate achieved by THz-NOMA transmission. This reduction is expected and can be explained in the following. Recall that  the $K$ beams, $\mathbf{f}_k$, are designed to match the primary users' channel vectors.  By increasing $N$, both the users' channel vectors and the $K$ spatial beams become more directional, which makes it more challenging for a secondary user to find a matching   beam. This performance degradation can be mitigated  if the beams are designed by taking both the primary and secondary users' channels into consideration. 

Recall that in this paper, each analog beamforming vector, $\tilde{\mathbf{f}}_k$, is selected from a codebook with the limited size ($N_Q$). Fig. \ref{fig5} is provided to show the impact of this important system parameter, $N_Q$, on the performance of THz-NOMA transmission. In particular, Fig. \ref{fig5} shows that the performance gain of THz-NOMA is larger by using a smaller $N_Q$, which can be explained in the following. The value of $N_Q$ decides the resolution of analog beamforming. For example, $N_Q\rightarrow \infty$ means the use of perfect analog beamforming and $\tilde{\mathbf{f}}_k$ will be perfectly matched to the channel vector of primary user ${\rm U}_k^P$. Therefore, using analog beamforming with finite resolution, i.e., $N_Q$ is small, provides an opportunity that the analog beamforming vector,  $\tilde{\mathbf{f}}_k$, might not be perfect for  primary user ${\rm U}_k^P$ but potentially ideal for some secondary users. This observation that  analog beamforming with finite  resolution is beneficial for the implementation of NOMA is also consistent to the findings previously reported in \cite{7918554}. 
%Recall that analog beamforming only offers the benefit of reduced hardware costs, because the implementation of digital beamforming requires multiple   radio frequency (RF) chains. Therefore,   the impact of using   hybrid beamforming and analog beamforming only is studied in Fig. \ref{fig6}.  As can be seen from the figure, removing digital beamforming does not cause much performance degradation to THz-NOMA transmission, particularly if the BB scheme is used. 

%    \begin{figure}[t]\centering \vspace{-2em}
%    \epsfig{file=AF.eps, width=0.5\textwidth, clip=}\vspace{-0.5em}
%\caption{ Comparison between hybrid beamforming and analog beamforming with THz-NOMA transmission.   $N=10$, $K=4$,$R_k=1$ BPCU, $N_Q=10$,  and $r_S=10$.
%  \vspace{-1em}    }\label{fig6}   \vspace{-2em} 
%\end{figure}

 \section{Conclusions}
 This paper has considered the use of    NOMA as an add-on in THz networks. In particular, it was assumed that there exists a legacy THz system, where     spatial beams have been configured  to serve   legacy primary users. The aim of this paper was to investigate how   these pre-configured  spatial beams can be used to serve additional secondary users without degrading the performance of the legacy network. The considered beam and power allocation problem was first formulated as   a mixed combinatorial non-convex optimization problem, and then solved by two methods, one based on the BB method and the other based on SCA. Both analytical and simulation results have been presented to demonstrate  that these spatial beams can be used as a type of bandwidth resources to connect additional users and yield a significant throughput gain. Note that this gain is achieved  without degrading the performance of the legacy system or acquiring additional spectrum. However, unlike conventional bandwidth resources,   spatial beams are     non-orthogonal resources, and hence it is important to study how to suppress inter-beam interference, which is an important direction for future research.    
 \appendices
 
 \section{Proof for Lemma \ref{lemma0}} \label{proof0}
The assumption that $\bar{R}_k^P\rightarrow 0$ means that $c_k\rightarrow -\infty$ and $b_{1k}\rightarrow -\infty$. Therefore,  $P^{\max}=\min\{P^{\max}, -c_k, -b_{1k}\}$ and   the constraints in \eqref{2tst:1}, \eqref{2tst:11}, and \eqref{2tst:2} are always satisfied. By using  this assumption,   problem \ref{pb:2} can be approximated   at high SNR as follows:  
  \begin{problem}\label{pb:2z} 
  \begin{alignat}{2}
\underset{\rho_{k}}{\rm{max}} &\quad    
   \sum^{K}_{k=1} s_k\log\left( 1+\frac{ h_{k}^S   \rho_{k}^S}{ \sum^{K}_{i=1,i\neq k} h_{i}^S s_i\rho_{i}^S  +  \rho^{P}\sum^{K}_{i=1,i\neq k} h_{i}^S   }
\right) \label{2ztobj:1} \\
\rm{s.t.} & \quad     \sum^{K}_{k=1}s_k=1, \quad s_k\in \{0,1\},   \quad \sum_{k=1}^{K}\rho_k^S\leq P^{\max}.
  \end{alignat}
\end{problem} 
where $\rho^P$ denotes the primary users' transmit power, and the notations, $h_{1k}^S$, $\rho_{1k}^S$, and $s_{1k}$, are simplified as  $h_{k}^S$, $\rho_{k}^S$, and $s_{k}$, respectively. 

For the considered special case, it is straightforward to show that the optimal solution of the greedy scheduling problem formulated in   \eqref{pb:2y} is simply given by
 $\rho^S_{k^*}=P^{\max}$, where $k^*=\arg \max \{h_{k}^S, 1\leq k \leq 2\}$.  Without loss of generality,   assume that the secondary user's effective channel gains on the two beams are ordered as follows: $h_1^S>     h_2^S$. Therefore, the key step to   prove that problems \ref{pb:2y} and   \ref{pb:2z}  have the same optimal solution is to show that assuming that $\rho_1^S+\rho_2^S=\rho^S\leq \rho^{\max}$, the optimal solutions of $\rho_1^S$ and $\rho_2^S$ are $\rho^{S}$ and $0$, respectively. Once this step is established, it is straightforward to show that the optimal value of $\rho^S$ is $\rho^{\max}$. 

  To simplify the proof, assume that $\alpha\rho^{S}$ is allocated to the first beam and the  $(1-\alpha)\rho^{S}$ is allocated to the second beam, $0\leq \alpha\leq 1$, which means that the objective function is given by
\begin{align}\label{obxx1}
 {f}(\alpha)\triangleq \log\left( 1+\frac{ h_{1}^S  \alpha \rho^{S}}{  h_{2}^S(1-\alpha)\rho^{S}  +    h_{2}^S \rho^{P}  }
\right)+ \log\left( 1+\frac{ h_{2}^S  (1-\alpha)\rho^{S}}{   h_{1}^S \alpha\rho^{S}  +    h_{1}^S  \rho^{P}  }
\right)  ,
\end{align} 
where $0\leq \alpha\leq 1$. The remainder of the proof is to show  that the objective function in \eqref{obxx1} is maximized by $\alpha=1$, regardless of the choices of $h_1^S$, $h_2^S$, $\rho^S$ and $\rho^P$.

Note that   $ {f}(\alpha)$ can be first rewritten  as follows:
\begin{align}
 f(\alpha)= &  \log\left( 1+\frac{ h_{1}^S  \alpha  }{  h_{2}^S(1-\alpha)   +    \beta h_{2}^S    }
\right)+ \log\left( 1+\frac{ h_{2}^S  (1-\alpha) }{   h_{1}^S \alpha  +  \beta  h_{1}^S   }
\right)   
\\\nonumber
= &  \log\left( 1+x\frac{   \alpha  }{   (1+\beta-\alpha)     }
\right)+ \log\left( 1+\frac{1 }{x}\frac{   (1-\alpha) }{     \alpha  +    \beta}
\right) , 
\end{align}
where $\beta=\frac{\rho^P}{\rho^S}$, $x=\frac{h_1^S}{h_2^S}$ and $x>1$ because of the assumption that $h_1^S>     h_2^S$.

    \begin{figure}[t]\centering \vspace{-0em}
    \epsfig{file=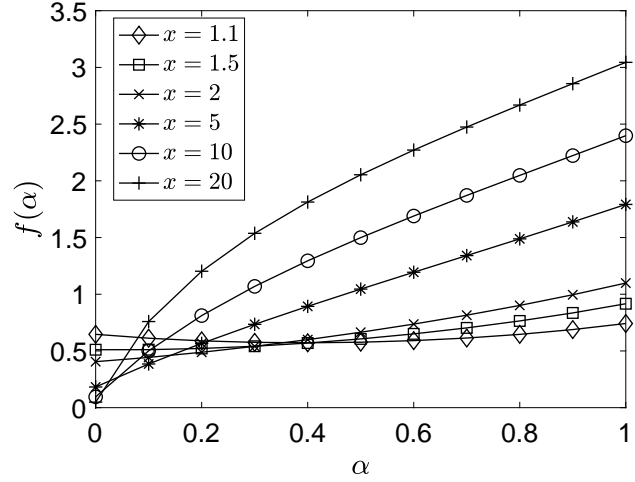, width=0.5\textwidth, clip=}\vspace{-0.5em}
\caption{ Property of the function $f(\alpha)$ with $\beta=1$. 
  \vspace{-1em}    }\label{fig00}   \vspace{-2em} 
\end{figure}

It is important to point out that $f(\alpha)$ can be either convex or concave, depending on the choice of $x$ and $\beta$, as shown in Fig. \ref{fig00}. However, $\alpha=1$ always maximizes $f(\alpha)$, regardless the choices of $x$ and $\beta$.   To show that $\alpha=1$   maximizes $f(\alpha)$,  the first order derivative of  $ f(\alpha)$
 with respect to $\alpha$ is first obtained as follows:  
  \begin{align}\nonumber
 &\frac{d f(\alpha)}{d\alpha}=  \frac{x-1}{ 1+\beta +(x-1)\alpha
}-\frac{1}{   \alpha-(1+\beta)
}+  \frac{x-1}{ (x-1)\alpha +x\beta+1}-\frac{1}{ \alpha+\beta  } 
\\\label{eq2333}&=   \frac{-  (1+\beta)\left( \left(x+1 \right)\alpha^2 +\left(x\left(\beta+\frac{x\beta+1}{x-1}\right) +\frac{1+\beta}{x-1}-(1+\beta) \right)\alpha+ 
\frac{x^2\beta^2+x\beta}{x-1} -\frac{(1+\beta)^2}{x-1}\right)
}{\left(\alpha +\frac{x\beta+1}{x-1}\right)(\alpha +\beta)  \left(\alpha+\frac{1+\beta}{x-1}\right) ( \alpha-(1+\beta))(x-1)} .
\end{align} 
Note that the quadratic function  in the numerator of \eqref{eq2333} is concave since $x>1$, $\alpha\leq 1$ and $\beta\geq 0$. 

%\begin{align}\nonumber
% \frac{d f(\alpha)}{d\alpha}=&   -2 \frac{   -x^2-x +4}{x+1  } = 2 \frac{   x^2+x -4}{x+1  } ,
%\end{align} 
%whose roots are $\frac{-1\pm \sqrt{17}}{2}$

As shown in Fig. \ref{fig00}, $f(\alpha)$ can be convex or concave, which makes the proof challenging. Interestingly,   $ \frac{d f(\alpha)}{d\alpha}$ is always positive at $\alpha=1$ for any choices of $x$ and $\beta$, as shown in the following:
\begin{align} \nonumber
\left.\frac{d f(\alpha)}{d\alpha}\right|_{\alpha=1}= &    \frac{ \left(x+1 \right)  + x\left(\beta+\frac{x\beta+1}{x-1}\right) +\frac{1+\beta}{x-1}-(1+\beta)   + 
\frac{x^2\beta^2+x\beta}{x-1} -\frac{(1+\beta)^2}{x-1}
}{\left(1 +\frac{x\beta+1}{x-1}\right)   \left(1+\frac{1+\beta}{x-1}\right) \beta (x-1)}  
\\\label{eq44x}
= &    \frac{ (1+\beta)^2x^2 -\beta x  -\beta^2 
}{\left(1 +\frac{x\beta+1}{x-1}\right)   \left(1+\frac{1+\beta}{x-1}\right) \beta (x-1)^2} .
\end{align} 
In order to show $\left.\frac{d f(\alpha)}{d\alpha}\right|_{\alpha=1}> 0$, it is sufficient to show that the numerator of \eqref{eq44x}, defined as $f_{\beta}(x)=(1+\beta)^2x^2 -\beta x  - \beta^2 $, is positive for $x>1$ and $\beta\geq 0$. Note that   the two roots of $f_{\beta}(x)=0$ are $\frac{\beta\pm \sqrt{\beta^2+4\beta^2(1+\beta)^2}}{2(1+\beta)^2}$. Therefore, the proof can be complete by showing that the positive root $\frac{\beta+ \sqrt{\beta^2+4\beta^2(1+\beta)^2}}{2(1+\beta)^2}< 1$, which   can be established  due to the equivalence between the following two inequalities: 
\begin{align}
 \sqrt{\beta^2+4\beta^2(1+\beta)^2} < 2(1+\beta)^2-\beta
\quad \Longleftrightarrow\quad
 \beta  < (1+2\beta).
\end{align}

The fact that $\left.\frac{d f(\alpha)}{d\alpha}\right|_{\alpha=1}>0$ is important because it shows that $ f(\alpha)$ is an increasing function at $\alpha=1$. Denote the two roots of the quadratic function in the numerator of \eqref{eq2333} by $r_1$ and $r_2$. Without loss of generality, assume $r_1\leq r_2$.  The use of  $\left.\frac{d f(\alpha)}{d\alpha}\right|_{\alpha=1}>0$ leads to the conclusion that $r_1\leq 1 \leq r_2$. Another important fact is that $f(0)<f(1)$, since scheduling the weak user   results in a smaller data rate  compared to the case with   the strong user scheduled. By using the two facts,  the proof for $\max f(\alpha) = f(1) $ can be established  as follows:

\subsubsection{ If  $r_1< 0$}  $\frac{d f(\alpha)}{d\alpha}$ is  positive for $0\leq \alpha\leq 1$, which means that $f(\alpha)$ is an increasing function for $0\leq \alpha\leq 1$. In Fig. \ref{fig00}, the   curves with $x=5$, $x=10$, $x=20$, $x=1.5$ and $x=2$ belong to this case.  Therefore, $\alpha=1$ can maximize the objective function $f(\alpha)$.

\subsubsection{ If  $r_1\geq  0$} $\frac{d f(\alpha)}{d\alpha}$ is  first non-positive for $0\leq \alpha<r_1$, and then becomes positive for $r_1\leq \alpha\leq 1$, which means that   $f(\alpha)$ is a non-increasing function for $0\leq \alpha<r_1$, and then becomes an increasing function $r_1\leq \alpha\leq 1$. In Fig. \ref{fig00}, the   curve with $x=1.1$  belongs to this case.  Furthermore, by using the fact that $f(0)<f(1)$,    $\alpha=1$ can still maximize the objective function $f(\alpha)$ in this case.

In summary, $\alpha=1$ can always maximize the objective function, and hence the proof for the lemma is complete.

\section{Proof for Lemma \ref{lemma1}}\label{proof1}
Recall that the aim of the tightening procedure is to find the maximal value of $\tilde{ {x}}_{\max,p}$, which can be achieved by solving the following optimization problem:
 \begin{problem}\label{pb:15x} 
  \begin{alignat}{2}
\underset{\mathbf{y} }{\rm{max}}  &\quad    
   \frac{ \mathbf{c}_p^T\mathbf{y}}{ \mathbf{d}_p^T\mathbf{R}\mathbf{y}  +t_{\mathbf{S}_{p1}\mathbf{S}_{p2}}} \label{15xtobj:1} \\
\rm{s.t.} & \quad    \mathbf{A}_p \mathbf{y}\leq \mathbf{b}_p   \label{15xttst:1}
\\&\quad \mathbf{A}_p^E \mathbf{y}= \mathbf{b}_p^E  .\label{15xttst:2} \end{alignat}
\end{problem}
It is important to point out   that the optimization variable vector,   $\mathbf{y}$, contains   $|\mathcal{S}|$ elements, whereas  \eqref{15xttst:2} contains $(|\mathcal{S}|-1)$ equality constraints. This important observation can be used to reduce the number of optimization variables from $|\mathcal{S}|$ to one only, as shown in the following.  

Without loss of generality, the constraint in \eqref{15xttst:2} can be expressed as follows: 
\begin{align}
\mathbf{a}_{p,p}^E  {y}_p +\tilde{\mathbf{A}}_p^E\tilde{\mathbf{y}}_p= \mathbf{b}_p^E.
\end{align}
  Assuming that   $\tilde{\mathbf{A}}_p^E$ is an invertible     matrix, the $(|\mathcal{S}|-1)$ optimization variables in $\tilde{\mathbf{y}}_p$ can be expressed as the following functions of $y_p$:
\begin{align}\label{ypp}
 \tilde{\mathbf{y}}_p=(\tilde{\mathbf{A}}_p^E)^{-1} (\mathbf{b}_p^E-\mathbf{a}_{p,p}^E  {y}_p).
\end{align}

By using  the fact that $\mathbf{c}_p$ is a $|\mathcal{S}|\times1$ all-zero vector except its $p$-th element, denoted by $ {c}_{p,p}$,  problem \ref{pb:15x} can be recast as the following optimization problem: 
 \begin{problem}\label{pb:167} 
  \begin{alignat}{2}
\underset{\mathbf{y} }{\rm{max}}  &\quad    
 \frac{  {c}_{p,p}{y_p}}{\bar{ {d}}_py_p+ \tilde{\mathbf{d}}\tilde{\mathbf{y}}_p  +t_{\mathbf{S}_{p1}\mathbf{S}_{p2}}} \label{167tobj:1} \\
\rm{s.t.} & \quad    \mathbf{a}_{p,p}y_p+ \tilde{\mathbf{A}}_p\tilde{\mathbf{y}}_p\leq \mathbf{b}_p   \label{167ttst:1} ,
\end{alignat}
\end{problem}
where $\bar{\mathbf{d}}= \mathbf{d}_p^T\mathbf{R}$, $\bar{ {d}}_p$ denotes the $p$-th element of $\bar{\mathbf{d}}$,  and  $\tilde{\mathbf{d}}$ is obtained from $\bar{\mathbf{d}}$ by removing $\bar{ {d}}_p$.

Furthermore, by using the fact that $(|\mathcal{S}|-1)$ optimization variables in $\tilde{\mathbf{y}}_p$ can be expressed as a function of $y_p$, constraint \eqref{167ttst:1} can be rewritten as follows: 
\begin{align}
\mathbf{a}_{p,p}y_p+ \tilde{\mathbf{A}}_p(\tilde{\mathbf{A}}_p^E)^{-1} (\mathbf{b}_p^E-\mathbf{a}_{p,p}^E  {y}_p)\leq \mathbf{b}_p ,
\end{align}
which can be further re-written as follows:
\begin{align}\label{newxxx1}
(\mathbf{a}_{p,p}-\tilde{\mathbf{A}}_p(\tilde{\mathbf{A}}_p^E)^{-1} \mathbf{a}_{p,p}^E ) {y}_p\leq (\mathbf{b}_p -\tilde{\mathbf{A}}_p(\tilde{\mathbf{A}}_p^E)^{-1} \mathbf{b}_p^E).
\end{align}
By using \eqref{newxxx1}, problem \ref{pb:16} can be expressed as the following simple optimization problem:
 \begin{problem}\label{pb:17x} 
  \begin{alignat}{2}
\underset{\mathbf{y} }{\rm{max}}  &\quad    
 \frac{  {c}_{p,p}{y_p}}{\bar{ {d}}_py_p+ \tilde{\mathbf{d}}\tilde{\mathbf{y}}_p  +t_{jk}} \label{17tobj:1} \\
\rm{s.t.} & \quad   {y}_p\leq \mathbf{a}_{\rm sign}\odot(\mathbf{b}_p -\tilde{\mathbf{A}}_p(\tilde{\mathbf{A}}_p^E)^{-1} \mathbf{b}_p^E)./ (\mathbf{a}_{p,p}-\tilde{\mathbf{A}}_p(\tilde{\mathbf{A}}_p^E)^{-1} \mathbf{a}_{p,p}^E ).
\end{alignat}
\end{problem}

Note that the following function,    $\frac{ax}{bx+1}$, is a monotonically increasing function of $x$ for $x\geq 0$, $a\geq 0$ and $b\geq 0$. Therefore, it is straightforward to show that the optimal solution of problem \ref{pb:17x} can be obtained as shown in the lemma.

\linespread{1.3}
\bibliographystyle{IEEEtran}
\bibliography{IEEEfull,trasfer}
  \end{document}